\documentclass[11pt]{article}
\usepackage{amsmath,amsfonts,amsthm,amssymb,color}
\usepackage{fancybox}
\usepackage{graphics}
\usepackage{graphicx}
\usepackage{tikz}
\usepackage{multirow,makecell}

\usepackage{hyperref,url}
\usepackage[T1]{fontenc}

\usepackage{algorithm,algorithmic}
\usepackage{mathtools}
\usepackage{chngcntr}
\usepackage{apptools}
\AtAppendix{\counterwithin{theorem}{subsection}}

\usepackage[margin=1in]{geometry}

\usepackage{placeins}

\makeatletter
\newtheorem*{rep@theorem}{\rep@title}
\newcommand{\newreptheorem}[2]{
\newenvironment{rep#1}[1]{
 \def\rep@title{#2 \ref{##1}}
 \begin{rep@theorem}}
 {\end{rep@theorem}}}
\makeatother

\newtheorem{theorem}{Theorem}[section]
\newreptheorem{theorem}{Theorem}

\newtheorem{lemma}[theorem]{Lemma}
\newreptheorem{lemma}{Lemma}
\newtheorem{observation}[theorem]{Observation}

\newtheorem{claim}[theorem]{Claim}
\newtheorem{fact}[theorem]{Fact}
\newreptheorem{fact}{Fact}

\theoremstyle{definition}
\newtheorem{definition}[theorem]{Definition}

\newenvironment{fminipage}
  {\begin{Sbox}\begin{minipage}}
  {\end{minipage}\end{Sbox}\fbox{\TheSbox}}

\def\pleq{\preccurlyeq}
\def\pgeq{\succcurlyeq}

\def\Span#1{\textbf{Span}\left(#1  \right)}

\def\defeq{\stackrel{\mathrm{def}}{=}}

\def\sgn#1{\mathrm{sgn} (#1)}

\def\union{\cup}

\def\abs#1{\left|#1  \right|}

\def\norm#1{\left\| #1 \right\|}

\def\calH{\mathcal{H}}
\def\calI{\mathcal{I}}

\def\calS{\mathcal{S}}
\def\calT{\mathcal{T}}

\newcommand\bb{\boldsymbol{\mathit{b}}}
\newcommand\cc{\boldsymbol{\mathit{c}}}
\newcommand\dd{\boldsymbol{\mathit{d}}}
\newcommand\ee{\boldsymbol{\mathit{e}}}
\newcommand\ff{\boldsymbol{\mathit{f}}}
\renewcommand\gg{\boldsymbol{\mathit{g}}}
\newcommand\hh{\boldsymbol{\mathit{h}}}

\let\muchless\ll

\newcommand\pp{\boldsymbol{\mathit{p}}}
\newcommand\qq{\boldsymbol{\mathit{q}}}

\newcommand\uu{\boldsymbol{\mathit{u}}}
\newcommand\vv{\boldsymbol{\mathit{v}}}
\newcommand\ww{\boldsymbol{\mathit{w}}}
\newcommand\yy{\boldsymbol{\mathit{y}}}
\newcommand\zz{\boldsymbol{\mathit{z}}}
\newcommand\xx{\boldsymbol{\mathit{x}}}

\renewcommand\AA{\boldsymbol{\mathit{A}}}
\newcommand\BB{\boldsymbol{\mathit{B}}}

\newcommand\FF{\boldsymbol{\mathit{F}}}

\newcommand\HH{\boldsymbol{\mathit{H}}}
\newcommand\II{\boldsymbol{\mathit{I}}}

\newcommand\MM{\boldsymbol{\mathit{M}}}
\newcommand\LL{\boldsymbol{\mathit{L}}}
\newcommand\PP{\boldsymbol{\mathit{P}}}
\newcommand\QQ{\boldsymbol{\mathit{Q}}}
\newcommand\RR{\boldsymbol{\mathit{R}}}

\newcommand\qqtilde{\boldsymbol{\mathit{\tilde{q}}}}

\newcommand{\eps}{\epsilon}

\newcommand{\trp}{\top}

\newcommand{\rank}{\text{rank}}

\DeclareMathOperator*{\nulls}{null}

\DeclareMathOperator*{\poly}{poly}

\newcommand{\schurto}[2]{\ensuremath{\textsc{Sc}\!\left[#1\right]_{#2}}}

\title{Incomplete Nested Dissection}

\author{
Rasmus Kyng
\thanks{email: \texttt{rjkyng@gmail.com}.
This work was done in part while the author was visiting the Simons
Institute for the Theory of Computing.
It was partially supported by the DIMACS/Simons Collaboration on Bridging Continuous and Discrete Optimization through NSF grant \#CCF-1740425.
}
\\ Simons Institute
\and
Richard Peng
\thanks{emails: \texttt{rpeng@cc.gatech.edu, \{rschwieterman3,pzhang60\}@gatech.edu}.
This material is based on work supported by the
National Science Foundation under Grant No. 1637566.}
\\ Georgia Tech
\and
Robert Schwieterman
\footnotemark[2]
\\ Georgia Tech
\and
Peng Zhang
\footnotemark[2]
\\ Georgia Tech
}

\newcommand{\esimple}{edge-simple}
\newcommand{\trussnice}{convex edge-simple}

\begin{document}

\maketitle

\thispagestyle{empty}

\begin{abstract}
We present an asymptotically faster algorithm for solving linear systems in 
well-structured 3-dimensional truss stiffness matrices. 
These linear systems arise from linear elasticity problems,
and can be viewed as extensions of graph Laplacians into higher dimensions.
Faster solvers for the 2-D variants of such systems
have been studied using generalizations of tools for solving graph
Laplacians [Daitch-Spielman CSC'07, Shklarski-Toledo SIMAX'08].

Given a 3-dimensional truss over $n$ vertices which is formed from a
union of  $k$ convex structures (tetrahedral meshes) with bounded aspect ratios, whose
individual tetrahedrons are also in some sense well-conditioned,
our algorithm solves a linear system in the associated stiffness matrix up to
accuracy $\eps$ in time $O(k^{1/3} n^{5/3} \log (1 / \eps))$.
This asymptotically improves the running time $O(n^2)$ by Nested
Dissection for all $k \ll n$.

We also give a result that improves on Nested Dissection even when we
allow any aspect ratio for each of the $k$ convex structures (but we
still require well-conditioned individual tetrahedrons). 
In this regime, we improve on Nested Dissection for $k \muchless n^{1/44}$.

The key idea of our algorithm is to combine
nested dissection
and support theory.
Both of these techniques for solving linear systems are well studied,
but usually separately.
Our algorithm decomposes a 3-dimensional truss into
separate and balanced regions with small boundaries.
We then bound the spectrum of each such region separately,
and utilize such bounds to obtain improved algorithms by
preconditioning with partial states of separator-based Gaussian elimination.
\end{abstract}

\clearpage
\setcounter{page}{1}

\section{Introduction}
\label{sec:intro}

Linear systems in truss stiffness matrices arise from linear elasticity problems for simulating the effect of forces on 
geometrically embedded objects~\cite{shklarskiT08,daitchS07}.
A truss is an undirected weighted graph over $n$ vertices which are points embedded in $d$ dimensions.
We refer to the weights as \emph{stiffness coefficients}.
The associated truss stiffness matrix is a $dn \times dn$ matrix, which can be written as a sum of rank-one positive semi-definite matrices such that each matrix corresponds to an edge in the truss graph.

Graph Laplacian matrices can be viewed as a special case of truss stiffness matrices, in which each vertex is embedded in 1 dimension.
2-dimensional variants have been studied in~\cite{shklarskiT08} and~\cite{daitchS07}. 
In this paper, we solve linear systems in truss stiffness matrices of 3-dimensional meshes.
These meshes commonly arise from applying finite element methods~\cite{strangF73} to 
simulating 3-dimensional physical space in scientific computing and numerical analysis.

Unlike linear systems in graph Laplacians for which nearly-linear time solvers exist (see for example~\cite{spielmanT14,KoutisMP10,KoutisMP11,pengS14,CKMPPRX14,kyngS16}, etc.), solving linear systems in general
truss stiffness matrices can be as hard as solving linear systems over the reals~\cite{kyngZ17}, even for 2-dimensional non-planar trusses.

Throughout this paper, we are interested in 3-dimensional trusses with additional geometric structures.
Specifically, we assume a truss is a tetrahedral mesh, which can be decomposed into several convex simplicial complexes.
In addition, each edge has bounded length and stiffness coefficient, and each tetrahedron has bounded aspect ratio.
The aspect ratio of a tetrahedron is the ratio between
its volume and the cubic of its diameter.
Bounded aspect ratio of simplices is a common criterion in mesh generation~\cite{bernEG94,chew89,mitchellV92,ruppert93}.

Existing works which use geometric structures to speed up algorithms are mostly for planar graphs 
(e.g.,~\cite{frederickson87,goodrich95,henzingerKRS97,
borradaileKMNW11}).
A common structure underlying these speed-ups is the $r$-division.
An $r$-division divides a graph into balanced pieces in which only few edges are between any two pieces~\cite{frederickson87,goodrich95,kleinMS13}.
Then divide-and-conquer can be applied.
Usually one layer of $r$-division is sufficient for running time
speedups~\cite{borradaileKMNW11,lkackiS11,borradaileENW14}, while recursive $r$-divisions 
lead to even faster algorithms~\cite{henzingerKRS97,BorradailePMNW11,kleinMS13}.
Designing efficient algorithms for 3-dimensional graphs is much harder than in
2-dimensions, even 
for graphs with 
bounded genus
surface embedded graphs~\cite{ericksonN11,borradaileENW14,borradaileCFN16}.

Existing linear system solvers for 3-dimensional tetrahedral meshes are
mainly based on nested dissection, 
which produces a vertex ordering for sparse Gaussian elimination~\cite{george73,liptonRT79,millerT90,alonY10}.
The main idea of nested dissection is to recursively compute small separators,  which recursively partition a (sub)graph into separate and balanced pieces.
Due to small overlaps between any two pieces, Gaussian elimination with this ordering introduces only a small fill-in size. 
The study of nested dissection, and the prevalence of linear systems
in 3-D complexes in turn motivated the study of small separators for
3-dimensional meshes~\cite{millerT90,millerTTV98}. 
However, the existence of 
3-D complexes with no 
sub-linear sized
constant-fraction separators~\cite{millerT90} means that such algorithms
critically depend on the aspect ratios of the individual tetrahedrons.

In addition, support theory has been widely used 
for solving linear systems in graph Laplacian matrices~\cite{spielmanT14,CKMPPRX14} and generalized Laplacian matrices (e.g., connection Laplacians~\cite{KLPSS16}).
Instead of directly solving the original linear system, support theory seeks a sparse preconditioner and iteratively solves a sequence of linear systems (refer to a survey~\cite{axelsson85}).
Both nested dissection and support theory are well-studied techniques in the literature of solving linear systems, but usually they are applied separately.

We design a new approach for solving linear systems in 3-D
truss stiffness matrices, by combining
the ideas of nested dissection and support theory.
We show that the $O(n^2)$-time\footnote{Here we
assume that multiplying two $n \times n$ matrices needs time $O(n^3)$, following the convention of nested dissection
literature~\cite{george73,liptonRT79,gilbertNP94}. We will discuss the running time by using fast matrix multiplication in Appendix~\ref{sec:rofl}.} 
bound for solving systems on simplicial
complexes with bounded aspect ratios~\cite{millerT90} can be further improved for
truss matrices on simplicial complexes
formed as a union of $k$ convex tetrahedral meshes.

\begin{theorem}[Informal statement]
Given a linear system in the stiffness matrix of a 3-D truss $\calT$
over $n$ vertices satisfying:
\begin{enumerate}
\item $\calT$ is a mesh of tetrahedrons formed by a union of $k$
  convex simplicial complexes with constant aspect ratio each,
\item each edge of $\calT$ has constant length and stiffness coefficient,
\item each tetrahedron of $\calT$ has constant aspect ratio, 
\end{enumerate}
and an error parameter $\eps > 0$, there is an algorithm which outputs a solution of the linear system 
up to accuracy $\eps$
in time $O(k^{1/3} n^{5/3} \log (1 / \eps))$.
\label{thm:mainInformal}
\end{theorem}
Theorem~\ref{thm:mainInformal} improves on Nested Dissection for all $k \muchless n $.
We also show a second result that improves on Nested Dissection even when we
allow any aspect ratio for each of the $k$ convex structures (but we
still require individual tetrahedrons to have bounded aspect ratio and size).
In this regime, we improve on Nested Dissection provided $k \muchless n^{1/44}$.

Our requirements for trusses have more restrictions than the
previous nested dissection based algorithms~\cite{millerT90},
which only requires constant aspect ratios for individual tetrahedrons.
Our additional restrictions are mainly due to the need to derive spectral
bounds for the preconditioner.
Assumptions such as bounded edge lengths and stiffness coefficients 
 (in addition to aspect ratios) 
are also present in previous work by Daitch and Spielman~\cite{daitchS07},
which forms the starting point for our key technical results on eigenvalues
of 3-D simplicial complexes.
Furthermore, the need to maintain null spaces as well as truss
structures in intermediate steps of Gaussian elimination places
additional requirements on the mesh structure.

Our algorithm can be viewed as incomplete nested dissection.
Specifically, for each convex simplicial complex, we compute an $r$-division whose boundaries have nice structures.
We then eliminate all interior vertices of the $r$-division, which gives a partial state of Gaussian elimination.
To solve the remaining linear system,
we simply use the $r$-division boundaries to precondition
 the intermediate submatrix of the Gaussian elimination as an
incomplete Cholesky factorization.
Our key technical result shows that the associated stiffness matrix of
 carefully chosen boundaries has bounded condition number, which may be of independent interest.

This algorithm only relies on the geometric structures of the underlying graphs and  good eigenvalue bounds of the associated matrices.
Our presentation 
only represents a basic
instantiation of utilizing both geometric structures and spectral ideas to solve linear systems.
It is likely that the same techniques can be extended to more linear systems.

Finally, due to reliance on numerical as well as combinatorial
structures, our result, as well as previous works on faster linear system
solvers for truss matrices~\cite{shklarskiT08,daitchS07}, are limited
to meshes with both bounded side lengths and aspect ratios.
This is significantly more restrictive than nested dissection based
algorithms that only depend on the aspect ratios of the truss elements.
We believe in future works it would be useful to extend our approach
to more general cases, and more systematically study the necessity
of restricting to such special cases.

\subsection{Related Works}
\label{subsec:Related}

\begin{table*}
\small
\centering
\begin{tabular}{|c|c|c|c|c|c|}
\hline
Algorithm &
D &
Global Req. &
Local Req. &
Runtime\\
\hline
Gaussian Elimination
&
any
&
any
&
none
&
$O(n^{3})$
\\
Fast Inversion~\cite{Legall14}
&
any
&
any
&
none
&
$O(n^{\omega})$
\\
\hline
Nested Dissection~\cite{liptonRT79}
&
2
&
any
&
none
&
$O(n^{\omega/2})$\\
Nested Dissection~\cite{millerT90}
&
3
&
any
&
aspect ratio (AR)
&
$O(n^{2\omega/3})$\\
\hline
Fretsaw Extension~\cite{shklarskiT08}
&
2
&
any 
&
none 
&
unspecified
\\
Augmented Tree~\cite{daitchS07}
&
2
&
stiffly-connected
&
size, AR, stiffness coefs.
&
about $n^{5/4} \log(1/\eps)$
\\
\hline
Theorem~\ref{thm:mainSmallAR}
&
3
&
\makecell{
$k$~$\times$~(convex \& AR)
}
&
size, AR, stiffness coefs.
&
$O(k^{1/3}n^{5/3}\log(1/\eps))$\\
Theorem~\ref{thm:main}
&
3
&
\makecell{
$k$~$\times$~convex
}
&
size, AR, stiffness coefs.
&
$O(k^{22/3}n^{11/6}\log(1/\eps))$\\
\hline
\end{tabular}
\caption{Comparisons of algorithms for solving linear systems
in stiffness matrices of trusses.
Global Req. / Local Req. refer to the restrictions / assumptions made
by these algorithms on the overall truss complex and individual truss
elements respectively.
Here $\omega$ is the matrix multiplication exponent,
which by~\cite{Legall14} is $< 2.3728639$.
The running time of~\cite{daitchS07} also ignores an overhead term of
$O(\log^{3/2}n \log\log^{3/4} n)$ related to the qualities of tree
embeddings.
}
\label{tbl:compare}
\end{table*}

A comparison of related results and their geometric requirements
are given in Table~\ref{tbl:compare}.
Many geometric restrictions are crucial to designing 
fast linear system solvers.
In addition to requirements on individual truss elements (e.g., bounded aspect ratios,  edge lengths and stiffness coefficients),  
a truss being \emph{stiffly connected}, defined in~\cite{daitchS07}, is important as a global requirement.
From a physical view, any deformation of a stiffly connected truss, except a translation and / or a rotation, requires energy.
From an algebraic view, 
the associated stiffness matrix of a stiffly connected truss only has a trivial null space.
In our algorithm, we always guarantee that our preconditioner of a convex truss is stiffly connected.

Both direct methods (mainly Gaussian elimination,
refer to~\cite{george73,liptonRT79,millerT90,gilbertNP94,Legall14}) and
iterative methods (e.g., conjugate gradient, multigrid methods, etc., refer to~\cite{fedorenko64,axelsson85,saad03})
have been widely studied for solving linear systems with additional structures.
Combining these two methods has led to fast linear system solvers in graph Laplacians and its generalizations (e.g.,~\cite{spielmanT14,KoutisMP10,KoutisMP11,pengS14,KLPSS16,CohenKPPRSV16,kyngS16} etc.).

One key idea among these algorithms is incomplete Cholesky factorization.
It was used in~\cite{kyngS16} to accelerate Gaussian elimination for graph Laplacian linear systems.
Each time the algorithm eliminates a vertex,
it sparsifies the partial Cholesky factorization by randomly picking a subset of its nonzeros.
In addition, sparsified block Cholesky factorization and multigrid methods were applied in~\cite{KLPSS16} to solving connection Laplacian linear systems.
In~\cite{KLPSS16}, they eliminate a subset of vertices, sparsify its Schur complement (which is an intermediate submatrix of Gaussian elimination, refer to Definition~\ref{def:schur} for a formal definition) without explicitly computing it, and then repeat this process recursively.

The analyses of both these two papers crucially rely on the fact that graph Laplacian matrices and connection Laplacian matrices are closed under taking the Schur complement.
However, this fact does not hold
 for truss stiffness matrices, for which the Schur complements can
be essentially any PSD matrices~\cite{kyngZ17}.
The hard instance given in~\cite{kyngZ17} is a 2D non-planar truss.
In this paper, instead, we utilize the geometric structures, balanced
divisions, and tetrahedral meshes whose tetrahedrons form simplicial
complexes and have bounded size and aspect ratio.
In this setting we can approximate Cholesky factorization.

One way of measuring the quality of a sparse approximated Cholesky factorization is the relative condition number of the sparse matrix and the original matrix, which determines the number of iterations of preconditioned iterative algorithms.
Cheeger's inequality provides tight bounds on the smallest nonzero eigenvalues of graph Laplacian matrices through sparse cuts in graphs~\cite{chung96}.
Extensions to higher-order eigenvalues, connection Laplacian matrices, and simplicial complexes can be found in~\cite{leeGT14,bandeiraSS13,steenbergenKM14,
parzanchevski2016isoperimetric}, etc.
However, it is unclear whether a Cheeger-type inequality exists for truss stiffness matrices.

\subsection{Organization of the Remaining Paper}

In Section~\ref{sec:prelim},
we give definitions and notations on graphs and linear algebra,
and formally define truss stiffness matrices.
In Section~\ref{sec:overview}, we present our algorithm for solving
linear systems in 3-dimensional trusses and prove our main theorems.
Sections~\ref{sec:minEig} and~\ref{sec:proofDS07Main} bound
eigenvalues of a well-structured truss stiffness matrix 
which is the key lemma for proving our main results.
In Section~\ref{sec:appendixHollow} we show our algorithm for constructing  preconditioners, and in
Section~\ref{sec:NestedDissection} we show  our nested dissection based algorithms for solving linear systems
in these preconditioners.
Appendix~\ref{sec:appendixSchur} gives some useful facts on Schur complements,
and Appendix~\ref{sec:rofl} gives our running times  parameterized
by $\omega$.

%!TEX root = main.tex

\section{Preliminaries}
\label{sec:prelim}

\subsection{Tetrahedral Meshes}

For a subset $S \subset \mathbb{R}^3$, we define the \emph{diameter} of $S$ to be the maximum Euclidean distance between any pair of points in $S$.
We define the \emph{aspect ratio} of $S$ to be the ratio between the radius of the smallest ball containing $S$ and the radius of the largest ball inscribed in $S$.
In mesh generation, aspect ratio is a common criterion for individual elements (see for example~\cite{bernEG94,chew89,millerT90,mitchellV92,ruppert93}).

A tetrahedron is the convex hull of four non-coplanar points in
$\mathbb{R}^3$.
We will specify tetrahedrons in terms of sets of four such points.
We refer to a set of tetrahedrons as a tetrahedral mesh.

We say a tetrahedral mesh is 
a \emph{simplicial complex} if the intersection of every two tetrahedrons is either empty or a face of both two tetrahedrons.
A simplicial complex is \emph{convex} if 
the union of the images of its simplices is convex, as defined in~\cite{CohenFMNPW14}.

\begin{definition}
\label{def:wellshapedmesh}
  A tetrahedral mesh is said to be \emph{simple} iff it is a
  simplicial complex and every tetrahedron has bounded aspect ratio.
\end{definition}

The following definition characterizes rigidity and stiffness of a tetrahedral mesh, 
which is an adaption of Definition 2.3 in~\cite{daitchS07}.
\begin{definition}
The \emph{rigidity graph} of a tetrahedral mesh is a graph whose vertices correspond to tetrahedrons and whose edges connect any two tetrahedrons sharing a triangle face.
A tetrahedral mesh is \emph{stiffly-connected} iff (1) its rigidity graph is connected, and (2) for any vertex in the mesh, the  rigidity subgraph induced on the tetrahedrons containing this vertex is connected.
\label{def:stifflyConnected}
\end{definition}

We define a \emph{bounding box} of a convex 3D shape to be a 3D box such that: (1) the box contains all points of this shape, and (2) the volume of the box is same as the volume of this shape up to a constant factor.
\cite{BarequetH01} gives a linear time algorithm which computes a bounding box of a convex shape in 3 dimensions.

\begin{lemma}(Lemma 3.6 of~\cite{BarequetH01})
\label{lem:BoundingBox}
Given a 3D convex shape, one can compute in linear time a bounding box of this shape.\footnote{The lemma statement in~\cite{BarequetH01} gives a bound related to the minimum-volume bounding box $B^*$ of the input shape, but their proof uses the volume of the shape as a lower bound of the volume of $B^*$.}
Moreover, the aspect ratio of the bounding box is same as the aspect ratio of the given shape up to a constant factor.
\end{lemma}

\subsection{Vectors and Matrices}

Given a vector $\xx \in \mathbb{R}^n$,
for $1 \le i < i+j \le n$, we denote $\xx_i$ the $i$th entry of $\xx$, and we denote $\xx_{i:i+j}$ the subvector whose entries are $\xx_i, \xx_{i+1}, \ldots, \xx_{i+j}$.
The Euclidean norm of $\xx$ is defined as $\norm{\xx}_2 \defeq \sum_{1\le i \le n} \xx_i^2$.

Given a matrix $\AA \in \mathbb{R}^{n \times n}$,
for $1\le i, j \le n$, we denote $\AA_{ij}$ the $(i,j)$th entry of $\AA$.
A square matrix $\AA \in \mathbb{R}^{n \times n}$ is a \emph{positive semi-definite matrix} (PSD) iff for every vector $\xx \in \mathbb{R}^n$ we have $\xx^\top \AA \xx \ge 0$.
We denote $\lambda_{\min}(\AA)$ the smallest \emph{nonzero} eigenvalue of $\AA$.

For two symmetric matrices $\AA, \BB \in \mathbb{R}^{n\times n}$,  we say $\AA \pgeq \BB$ iff $\AA - \BB$ is  PSD.
We define the 
\emph{condition number} of $\AA$ relative to $\BB$, denoted by $\kappa(\AA, \BB)$, to be 
\[
\min  \left\{ \frac{\lambda_{\max}}{\lambda_{\min}}: \lambda_{\min} \cdot \BB \pleq \AA \pleq \lambda_{\max} \cdot \BB \right\}.
\]

In addition, we define Schur complements which arise from the process of Gaussian elimination.
\begin{definition}[Schur complement]
\label{def:schur}
Let $S,T$ be a partition of the indices of a square matrix $\AA$ so that
$\AA = \left(\begin{array}{cc}
\AA_{SS} & \AA_{ST} \\
\AA_{ST}^{\trp} & \AA_{TT}
\end{array} \right)$ where $\AA_{SS}, \AA_{ST}, \AA_{TT}$ are block
matrices, 
the \emph{Schur complement} of $\AA$ onto $T$ is 
\[
\schurto{\AA}{T} \defeq \AA_{TT} - \AA_{ST}^\top\AA^{-1}_{SS}\AA_{ST}.
\]
\end{definition}
We will use the following fact of Schur complements.
\begin{fact}[Lemma 4 of~\cite{roseTL76}]
Let $\AA = \left( \begin{array}{cc}
\AA_{SS} & \AA_{ST} \\
\AA_{ST}^\top & \AA_{TT}
\end{array} \right)$ be a symmetric matrix, the $(i,j)$th entry of the Schur complement $\textsc{Sc}[\AA]_T \neq 0$ only if there exists a sequence of indices $k_1, \ldots, k_l \in S$ such that all $\AA_{ik_1}, \AA_{k_1k_2}, \ldots, \AA_{k_{l-1}k_l}, \AA_{k_l, j}$ are nonzero.
\label{fac:SchurPath}
\end{fact}

\begin{fact}
Let $\AA$ be a symmetric PSD matrix and $\schurto{\AA}{T}$ be its Schur complement.
\begin{enumerate}
\item $\schurto{\AA}{T}$ is a symmetric PSD matrix.
\item $\lambda_{\max}(\schurto{\AA}{T}) \le \lambda_{\max}(\AA)$.
\end{enumerate}
\label{fac:Schur}
\end{fact}

We prove Fact~\ref{fac:Schur} in Appendix~\ref{sec:appendixSchur}.

\subsection{Solving Linear Systems}
\label{sec:prelim:linsolve}

Our algorithm combines two of the most important tools
for solving linear systems: Nested dissection and preconditioning.
Below, we give a brief introduction to some of the central results on
these techniques.

Classic results due to Lipton, Rose, and Tarjan~\cite{liptonRT79}, and
Miller and Thurston~\cite{millerT90} combine to show that linear
systems arising from 
simple tetrahedral meshes (see
Definition~\ref{def:wellshapedmesh}) 
can be solved in $O(n^2)$ time.
These results concern linear equations in an $n\times n$ matrix $\AA$ where
the indices $\{1,\ldots, n\}$ can be embedded as points $\{\pp_{1},
\ldots, \pp_{n}\}$ that form the vertices of an explicitly given, simple
tetrahedral mesh, and $\AA_{ij}$ is non-zero only if the vertices $i$
and $j$ share an edge in the tetrahedral mesh.

\begin{theorem}[Nested dissection~\cite{millerT90}]
Let $\AA \in \mathbb{R}^{n \times n}$ be a symmetric matrix defined on a simple tetrahedral mesh.
A Cholesky factorization $\AA = \PP \LL \LL^\top \PP^\top$ can be computed in time $O(n^{2})$, in which $\PP$ is a permutation matrix and $\LL$ is a lower triangular matrix with $O(n^{4/3})$ nonzero entries.
As a result, a linear system in $\AA$ can be solved in time $O(n^{2})$ by Gaussian elimination.
\label{thm:nestedDissection}
\end{theorem}

Theorem~\ref{thm:nestedDissection} can be extended to a block matrix $\AA \in \mathbb{R}^{cn \times cn}$ where $c$ is a constant positive  integer.
Each vertex of the underlying graph corresponds to $c$ indices of $\AA$.
In addition, the block corresponding to the column
indices for vertex $i$ and the row indices for $j$ should be non-zero only if the vertices $i$
and $j$ share an edge in the tetrahedral mesh, or if $i = j$, i.e. when the block is on the diagonal.

The Nested Dissection algorithm relies on invoking \emph{separators}
recursively.
A separator is a set of indices $S$ such that the remaining indices $[n] \setminus S$
can be partitioned into two sets $B$ and $C$ such that every entry with $i \in B$ and $j \in C$ has $\AA_{ij} = 0$.
Furthermore, we guarantee that the partition is roughly balanced, for example, each of $B$ and $C$ contains no more than $\frac{3}{4} \cdot n$ indices.
Nested Dissection recursively repeats the partitioning process on the union
of each subset and the separator itself, that is, $B \cup S$ and $C \cup S$.
Given such a recursive partition scheme, we reorder the indices of the matrix so that the indices in the separator $S$ are eliminated \emph{last}, and we then order the indices in $B$ and $C$ recursively  in a similar way.
We perform Gaussian elimination on the matrix according to this ordering, which only introduces a small fill-in size and few multiplication counts.
This approach also works for eliminating a subset of
the variables, resulting in a Schur complement on the rest.

Both the running time and representation cost of 
nested dissection algorithms are bottlenecked by the costs
of the top-level separators.
In Algorithm~\ref{alg:TrussSolver} $\textsc{TrussSolver}$,
we will utilize improved running time bounds for nested dissection when
better separators exist.
The following lemma characterizes the performance of Nested Dissection
given better top-level separators.

\begin{lemma}
Suppose we have a recursive separator decomposition
of a 
simplicial complex with $n$ bounded aspect ratio
tetrahedrons such that:
\begin{enumerate}
\item the number of leaves, and hence total number
of recursive calls, is at most $n^{\alpha}$.
\item each leaf (bottom layer partition) has at most
$n^{\beta}$ tetrahedrons.
\item each top separator has size at most $n^{\gamma}$.
\end{enumerate}
Then we can find an exact Cholesky factorization of
the associated stiffness matrix in time
$O(n^{\alpha + 2 \beta}   + n^{\alpha + 3 \gamma})$,
and the total resulting fill-in is
$O(n^{\alpha + \frac{4}{3} \beta}  + n^{\alpha + 2 \gamma})$.
\label{lem:flexibleND}
\end{lemma}

We give a proof of the 
above lemma in Section~\ref{sec:flexibleNDProof}, which is an adaption of the analysis in~\cite{liptonRT79}
We remark that the algorithmic realization of this can be viewed
as utilizing the nested dissection algorithm~\ref{thm:nestedDissection}
to complete this structure into a full separator tree.

Last but not  least, we state the following theorem for preconditioned conjugate gradient, which will be used in bounding the running time of our algorithm.

\begin{theorem}[Preconditioned conjugate gradient~\cite{axelsson85}]
Let $\AA, \BB \in \mathbb{R}^{n \times n}$ be two symmetric positive semidefinite matrices and let $\bb \in \mathbb{R}^n$.
Each iteration of the preconditioned conjugate gradient multiplies one vector by $\AA$, solves one linear system in $\BB$, and performs a constant number of vector additions. 
For any $\epsilon > 0$, the algorithm outputs 
an $\xx$ satisfying
$\norm{\AA \xx - \bb}_2 \le \eps \norm{\bb}_2$ in $O(\sqrt{\kappa(\AA, \BB)} \log(1/\epsilon))$ such iterations.
\label{thm:pcg}
\end{theorem}

We remark that while there are settings where the convergence
of preconditioned conjugate gradient is numerically unstable,
the eigenvalue-based bound that we utilize here is stable
once the solves involving $\BB$ have polynomially small errors.

\subsection{Truss Stiffness Matrices}

We extend the definition of 2-dimensional truss stiffness matrices from~\cite{daitchS07} (see Definition 2.1 and 2.2) to 3 dimensions.

\begin{definition}[3-dimensional truss]
\emph{3-dimensional truss} $\mathcal{T} = \langle V, \{ \pp_i \}_{i \in V}, T, E, \gamma \rangle$ is given by
\begin{itemize}
\item  A set of $n$ vertices $V$ embedded at distinct points $\pp_1, \ldots,
\pp_n \in \mathbb{R}^3$
\item A mesh (i.e. set) of tetrahedrons $T = \{t_1, t_2, \ldots\}$,
  each specified in terms of four vertices,
  i.e. we identify tetrahedron $t_i$ with both four vertices
  $\{a_{i}, b_{i}, c_{i}, d_{i} \} \subseteq V$ and the convex hull of $\pp_{a_{i}}, \pp_{b_{i}},
    \pp_{c_{i}}, \pp_{d_{i}}$.
 \item A set of edges $E$ which is exactly the set of pairs of
   vertices that appear in some tetrahedron together. 
Each edge $e = (i,j) \in E$ represents a straight idealized bar
between vertex points $\pp_i$ and $\pp_j$.
 \item A function $\gamma: E \to \mathbb{R}_+ $, which assigns a
   stiffness coefficient $\gamma(e)$ to each edge $e$. 
The stiffness coefficient represents the stiffness of the idealized
bar corresponding to edge $e$.
\end{itemize}
\end{definition}

\begin{definition}[Truss stiffness matrix]
Let $\calT =\langle V, \{ \pp_i \}_{i \in V}, T, E, \gamma \rangle$ be a 3-dimensional truss.
For each edge $e = (i,j) \in E$, we define an edge vector $\bb^{(e)} \in \mathbb{R}^{3n}$
with 6 nonzero entries:
\[
\bb^{(e)}_{3i-2:3i} = - \bb^{(e)}_{3j-2:3j} = \frac{\pp_i - \pp_j}{\norm{\pp_i - \pp_j}_2}.
\]
The \emph{stiffness matrix} of the truss $\calT$ is defined as
\[
\AA_{\calT} \defeq \sum_{e =(i,j) \in E} \frac{\gamma(e)}{\norm{\pp_i - \pp_j}_2 } \bb^{(e)}\bb^{(e)\top}.
\]
\label{def:truss}
\end{definition}

In general, solving linear systems in truss stiffness matrices can be
as hard as solving linear systems in real matrices~\cite{kyngZ17}.

In this paper, we study 3D trusses with some additional geometric structures. 
These structures enable us to design linear system solvers that run much faster than solvers for general simple tetrahedral meshes.

\begin{definition}
  We say a 3D truss is \emph{\esimple} if its tetrahedral mesh is
  simple and every tetrahedron has bounded edge
  lengths and stiffness coefficients (i.e. both are bounded above and
  below by constants).
\end{definition}

\begin{definition}
  We say a 3D truss is \emph{\trussnice}
  if it is {\esimple}, and its tetrahedral mesh is convex.
\end{definition}

\section{Algorithm Overview}
\label{sec:overview}

In this section, we present our algorithm for solving linear systems
in stiffness matrices of {\esimple} 3D trusses.
Our first main result concerns trusses formed by combining $k$ {\trussnice} trusses, each with constant aspect ratio upper bounded by some
arbitrarily large but fixed constant.

\begin{theorem}
Given an {\esimple} 3-D truss with $n$ vertices, formed from a union of $k$
{\trussnice} trusses each with aspect ratio at most $O(1)$\footnote{A slightly modified analysis extends this result to allow each individual truss has aspect ratio $O(n_i^{1/4})$, where $n_i$ is the number of vertices of the $i$th individual truss.}, and an error parameter $\eps >0$,
there is an algorithm which solves 
a linear system in the corresponding stiffness
matrix up to accuracy $\eps$ in time $O(k^{1/3} n^{5/3} \log (1 / \eps))$.
\label{thm:mainSmallAR}
\end{theorem}

This theorem is appealing because in many modeling
applications, only large constant aspect ratios are needed for
individual convex parts that are being combined.
In Theorem~\ref{thm:mainSmallAR}, we see that the performance degrades
smoothly towards the $O(n^2)$ running time of Nested Dissection as $k$
approaches $n$.

Our second main result deals with the case when we allow a truss formed
from $k$ {\trussnice} trusses, each of which may have arbitrarily large
aspect ratio.

\begin{theorem}
Given an {\esimple} 3-D truss with $n$ vertices, formed from a union of $k$ {\trussnice} trusses, and an error parameter $\eps > 0$, 
there is an algorithm which solves
a linear system in the corresponding stiffness
matrix up to accuracy $\eps$ in time 
$O(n^{11/6} k^{22/3} \log (1 / \eps))$.
\label{thm:main}
\end{theorem}

We remark that the geometric assumptions of Theorem~\ref{thm:main}
are in many
ways fairly weak. We need the individual truss tetrahedrons
to have small aspect ratio, but each of the $k$ {\trussnice} trusses may
overall have a wide range of shapes: It can form a ball, a pancake, or
even a very long beam with arbitrarily large aspect ratio.
The dependence on $k$ is fairly bad and has not been carefully optimized, meaning currently that only
about $k \muchless n^{1/44} \approx n^{0.0227} $
{\trussnice} trusses can be combined while still
achieving a speed-up over nested dissection.
Even so, this allows the construction of some shapes with genus up to $n^{1/44}$. 

Fast matrix multiplication can be used in our algorithm, as well as earlier routines.
In accordance with previous works on nested dissection,
and to simplify presentation, we assume $\omega = 3$ (the matrix multiplication constant)
throughout our calculations.
However, in Appendix~\ref{sec:rofl},
 we also give the
$\omega$-dependent bounds.
Assuming $\omega = 2.3728639$ as in~\cite{Legall14},
the bounded aspect ratio case from Theorem~\ref{thm:mainSmallAR} takes time
$O(k^{0.1210452} n^{1.4608641} \log(1 / \epsilon) )$,
while the arbitrary aspect ratio case from Theorem~\ref{thm:main} takes time
$O(k^{5.7115596} n^{1.5175803} \log (1/\eps))$.
In both cases the running times are less than the $O(n^{1.5819093})$
bound obtained by plugging $\omega = 2.3728639$ into 3-D
nested dissection.

Our algorithm pseudocode
is stated
in Algorithm~\ref{alg:TrussSolver}.
Note this algorithm proves both Theorem~\ref{thm:mainSmallAR} and~\ref{thm:main}. 
Theorem~\ref{thm:mainSmallAR} is a special case where $\calI = [k]$ in line~\ref{line:MainI}.

\subsection{Main Ideas}

Both our main theorems are based on speeding up Nested Dissection by
combining it with preconditioning and iterative solvers.
In Section~\ref{sec:prelim:linsolve} we gave a brief outline of Nested Dissection.
In classical Nested Dissection, the main bottleneck that constrains the running
time of the algorithm is the process of applying Gaussian elimination
to the few separators in the top levels of the separator tree, after having
eliminated all the matrix indices at lower levels.
The intermediate matrix that arises during Nested Dissection after
eliminating the indices at lower levels is in fact the Schur
complement onto the separators at top levels.

Our central idea that gives us an advantage over Nested Dissection is
that: 
the outer boundary of tetrahedrons of a single {\trussnice} truss\footnote{It can be a subset of one of the $k$ input individual {\trussnice}trusses.}
with small aspect ratio is a 
good preconditioner for the Schur
complement of the whole truss onto the boundary.

By forcing Nested Dissection to use these outer boundaries of
individual {\trussnice} trusses as the top-level separator, 
we get a separator which has a good sparse preconditioner.
Fortunately, we can also ensure that this top-level separator has a small size.

The phenomenon that the boundary itself is a good preconditioner for the
Schur complement onto the boundary has a natural interpretation based on structural
mechanics.
The quadratic form associated with the Schur complement
corresponds to the energy associated with deforming the whole truss by
squishing or stretching the boundary vertices while leaving the interior intact and
finding the positions of the interior vertices that minimize the
overall energy.
We show that this energy is not much more than the energy that arises
from applying the same deformation to just the boundary tetrahedrons
after deleting the interior vertices.

This means that we can speed up the process of applying the inverse of
the Schur complement onto the boundary, 
which is the bottleneck of nested dissection.
We avoid directly inverting the Schur complement by
instead solving a linear system in the Schur complement using the
boundary itself as a preconditioner in preconditioned conjugate
gradient.
Because the boundary is much sparser than the Schur complement, and
has much fewer tetrahedrons than the initial mesh, we significantly reduce the running time.

How good a preconditioner the boundary is for the Schur complement
depends on the number of tetrahedrons in the initial convex mesh.
If the mesh is large, then the preconditioner is worse 
in the sense that solving an associated linear system is time-consuming, 
but  the
gain from hollowing out the mesh is relatively larger, because the
number of tetrahedrons on the surface is relatively smaller.
Ideally, we want to balance these two phenomena against each other.
We can ensure a good trade-off between these two effects by first dividing very large
convex trusses into smaller chunks before hollowing out these chunks
and using the resulting boundaries, which now look somewhat like a
Swiss cheese, as a preconditioner.

All together, this approach of partitioning, hollowing out,
preconditioning, and using Nested Dissection gives us Theorem~\ref{thm:mainSmallAR}.

If the aspect ratios of individual convex trusses are allowed to be extremely
large, so that the trusses can be very thin, then we cannot gain much by
hollowing out the trusses.
However, if an individual convex truss has large aspect ratio,
then we can get good separators for
Nested Dissection
by slicing the truss along its longest dimension.
Combining this observation with our preconditioning approach, we
are able to obtain Theorem~\ref{thm:main}, which has no requirements
on the aspect ratios of each of the $k$ trusses we are combining.
Unfortunately, leveraging both the preconditioning behavior and the
existence of good separators for individual large aspect ratio trusses
requires fairly technical work, which currently introduces a bad
dependence on $k$ in this version of our main result.

\begin{algorithm}[htb]
\caption{
  $\textsc{TrussSolver}(
\calT = \langle V, \{ \pp_i \}_{i \in V}, T, E, \gamma \rangle,
\ff,
\epsilon,
c_\alpha, c_r)$
  \label{alg:TrussSolver}
}
\begin{algorithmic}[1]
\renewcommand{\algorithmicrequire}{\textbf{Input:}}
\renewcommand\algorithmicensure {\textbf{Output:}}
\REQUIRE{a 3D truss $\calT = \langle V, \{ \pp_i \}_{i \in V}, T, E, \gamma \rangle$ with $n$ vertices, which is a union of $k$ {\trussnice}trusses $\calT_1, \ldots, \calT_k$,
a vector $\ff \in \mathbb{R}^{3n}$,
an error parameter $\epsilon > 0$.\\
Constants for aspect ratio threshold $0 < c_{\alpha} < 1$,
  and hollowing rate $0 < c_{r} < 1$.}
\ENSURE{an approximate solution $\xx$ such that
$\norm{\AA_{\calT} \xx - \ff}_2 \le \eps \norm{ \ff}_2$.}
\FOR{each $i$}
  \STATE{
    Compute a bounding box $B_i$ of $\calT_i$, via Lemma~\ref{lem:BoundingBox}.
    \label{line:BoundingBox}
  }
\ENDFOR
\STATE{Let $\calI = \{1 \le i \le k: \alpha(\calT_i) \le n_i^{c_{\alpha}} \}$.
  \label{line:MainI}}
 \FOR{each $i \in \calI$}
        \STATE{Hollow out the interior vertices of $\calT_{i}$
		with parameter $r_{i} = n_{i}^{c_r}$ to form $\calH_{i}$.
             \label{line:MainHollow}}
\ENDFOR
\STATE{Run nested dissection on the preconditioner (possibly with a specific set of separators).
  \label{line:MainND}}
\STATE{Run preconditioned conjugate gradient with this preconditioner to solve the overall system.}
\RETURN the solution $\xx$.
\end{algorithmic}
\end{algorithm}

\subsection{Bounding Eigenvalues of an Edge-Simple and Stiffly-Connected Truss}

Our main structural results are bounds on the condition number
of the stiffness matrix of a
{\esimple} and stiffly connected (see Definition~\ref{def:stifflyConnected}) simplicial complex.
As each vertex is involved in at most a constant number of
tetrahedrons, we can easily obtain bounds on the
maximum eigenvalue.

Thus, our main technical contribution is a lower bound
on the minimum non-zero eigenvalue of the stiffness matrix
of a 3D {\esimple} truss.
We state the explicit bound in Lemma~\ref{lem:3dMinEig}.
Such a bound is analogous to the bound on minimum eigenvalues
of a path in a graph.

\begin{lemma}
Let $\calT$ be an {\esimple} and stiffly-connected 3D truss. 
Let $n$ be the number of vertices of $\calT$ and $\Delta$ be the diameter.
Let $\MM$ denote the associated stiffness matrix. 
Then, $\lambda_{\min}(\MM) = \Omega(n^{-1} \Delta^{-4})$ and $\rank(\MM) = 3n - 6$.
\label{lem:3dMinEig}
\label{lem:TrusssLambdaMin}
\end{lemma}

Our proof is heavily motivated by the Path Lemma by Daitch and Spielman~\cite{daitchS07}, which bounds the minimum non-zero eigenvalue of the stiffness matrix of a path of triangles in 2D.
We start by shifting all the tetrahedrons by the
normals w.r.t. a particular centering tetrahedron, which
in effect projects away the coordinates from the null space.
Then we lower bound the minimum dot-product of a quadratic
form in terms of the pairwise $\ell_{2}$ differences
among these tetrahedrons.
However, our calculations 
result in an exponential factor loss depending on the ``hop
distance'' between the first and last tetrahedrons.
This is due to the accumulation of rotational operators,
which we need to treat as matrices instead of simple rotations.

We are not sure whether this exponential increase is simply
due to an algebraic artifact in our proof.
To circumvent it, we instead show that there exists
a particular centering where a pair of close-by
tetrahedrons are placed $1 / \poly(n)$ apart. 
This proof relies on analyzing the ``average behavior''
of all centerings globally.
It once again relies on treating the initial rotations
and projections as linear operators, and working
directly with the singular values of these projection
operators.

\subsection{Proving the Main Result for Small
  Aspect Ratio Truss Unions}

To accelerate Nested Dissection, we need to find a set of balanced
separators that are small and whose Schur complements have good sparse preconditioners.

To build these good separators, motivated by $r$-divisions, we divide each of the small-aspect
ratio trusses in our union of $k$ trusses into smaller chunks, such
that the boundary of each chunk is a good preconditioner of the Schur
complement onto that boundary.
The union of all these boundaries is called a \emph{hollowing}.
To create a hollowing, we  fix two parameters: 
a bounding box $B$ that determines the directions of each smaller chunks of the hollowing, and
a size parameter
$r$ that controls the size of the smaller chunks.
We call each smaller chunk as a \emph{region}.

\begin{definition}[$(B, r)$-hollowing]
Given a {\trussnice}3D truss
$\calT = \langle V, \{ \pp_i \}_{i \in V}, T, E, \gamma \rangle$, a bounding box $B$ of $\calT$, and
a parameter $r \le n / \alpha^2$ where $\alpha$ is the aspect ratio of $\calT$,
a \emph{$(B, r)$-hollowing} of $\calT$ is
another {\esimple}
3D truss
${\calH = \langle U, \{ \pp_i \}_{i \in U}, S, F, \gamma' \rangle}$
such that, $U \subseteq V$, $S \subseteq T$, and $F$ is the subset of $E$
that arises from edges in $S$, while $\gamma'$ is just the restriction of
$\gamma$ to $F$. I.e. edges maintain the same stiffness factors as in $\calT$.
Also
  \begin{enumerate}
  \item  $\calH$ contains $O(n r^{-1/3})$ points.
$\calT \setminus \calH$ consists of $O(nr^{-1})$ disjoint chunks, each of which has $O(r)$ vertices and is incident to $O(r^{2/3})$ vertices of $\calH$.
  \item  for \emph{every} plane $P$ whose normal vector has angle $\theta \in (0, \pi/2)$ with the longest direction of $B$, 
  the number of tetrahedrons in $\calH$ intersected by $P$ 
  is
  \[
  O\left( n^{2/3} \alpha^{-1/3} r^{-1/3} \cos^{-2}\theta \right).
  \]
  \item $\AA_{\calH} \preceq \schurto{\AA_{\calT}}{U} \preceq O(r^{2}) \AA_{\calH}$.
  \end{enumerate}
  \label{def:hollow}
\end{definition}

The next lemma describes the performance of  algorithm \textsc{Hollow}, Algorithm~\ref{alg:hollow} in Section~\ref{sec:appendixHollow}, that we use
to compute a $(B, r)$-hollowing of a {\trussnice} truss.

\begin{lemma}
Given a {\trussnice} 3-dimensional truss $\calT = \langle V, \{ \pp_i \}_{i \in V}, T,
E, \gamma \rangle$ with $n$ vertices,  a bounding box $B$ of $\calT$, and
a positive integer $r$ such that the aspect ratio of $\calT$ is at most $\sqrt{n/r}$,
the algorithm $\textsc{Hollow}(\calT,B, r)$ returns a
$(B, r)$-hollowing $\calH$ of $\calT$,
and runs in time $O(n)$. 
\label{lem:hollow}
\end{lemma}

\begin{proof}[Proof of Theorem~\ref{thm:mainSmallAR}]
Let $\calT$ be a {\esimple} 3-D truss with $n$ vertices, formed from a union of $k$
{\trussnice} trusses, say $\calT_1, \ldots, \calT_k$, each with aspect ratio at most $O(1)$.
For each $1\le i \le k$, let $n_i$ be the number of vertices of $\calT_i$, and define $r_i \defeq n_i^{1/2}$.
In Algorithm~\ref{alg:TrussSolver}, for each $\calT_i$,
 we compute a $(B_i, r_i)$-hollowing, where $B_i$ is a bounding box of $\calT_i$.
 By Lemma~\ref{lem:BoundingBox} and~\ref{lem:hollow}, the total running time here is $O(n)$.
In each $(B_i, r_i)$-hollowing region, we eliminate its interior vertices in total time
\[
O\left(\sum_i n_ir_i^{-1} \cdot  r_i^2\right)
=
O\left(n^{3/2}\right).
\]
The Schur complement onto the boundaries has 
\[
O\left(\sum_i n_ir_i^{-1} \cdot (r_i^{2/3})^2\right)
=
O\left(n^{7/6}\right)
\]
nonzeros.
We then run preconditioned conjugate gradient (PCG) to
solve the linear system in the Schur complement by preconditioning
it via the union of the $(B_i, r_i)$-hollowings, say $\calT'$.
Note by Jensen's inequality, $\calT'$ has size
\[
O\left(\sum_i n_i r_i^{-1/3} \right)
=
O\left(\sum_i n_i^{5/6}\right)
=
O\left(k^{1/6} n^{5/6}\right).
\]

Before running PCG, we compute a Cholesky factorization of $\AA_{\calT'}$ by nested dissection. According to Theorem~\ref{thm:nestedDissection}, the running time is $O(k^{1/3}n^{5/3})$, and the fill-in size is $O(k^{2/9}n^{10/9})$.
By Definition~\ref{def:hollow}, the condition number is $O(\max_i r_i^2) = O(n)$. According to Theorem~\ref{thm:pcg}, the number of PCG iterations is at most $O( n^{1/2} \log (1/\eps))$ to output a solution up to accuracy $\eps$.
In each PCG iteration, we do a matrix-vector multiplication with the Schur complement in time $O(n^{7/6})$, and solve a linear system in $\AA_{\calT'}$ in time $O(k^{2/9}n^{10/9})$. 
Thus the total running time is $O(k^{1/3}n^{5/3} \log (1 / \eps))$.
\end{proof}

\subsection{Proving the Main Result for All-Aspect Ratio Truss Unions}

We extend our result to cover the case when the union of 
{\trussnice} trusses also include trusses with arbitrarily large aspect
ratios.
The lemma below shows that  large aspect
ratio implies the existence of good plane separators, which is proven
in Section~\ref{sec:NestedDissection}.

\begin{lemma}
Given a {\trussnice} 3D truss with aspect ratio at least $\alpha > 0$, say,
$\calT =\langle V, \{ \pp_i \}_{i \in V}, T, E, \gamma \rangle$, and its bounding box $B$.
Let $\dd \in \mathbb{R}^3$ be a unit vector along the longest direction of $B$, and let $\gg \in \mathbb{R}^3$ be a unit vector with $\dd \cdot \gg > 0$.
Then every plane orthogonal to $\gg$ intersects at most $O(n^{2/3} \alpha^{-1/3}  (\dd \cdot \gg)^{-1})$ tetrahedrons.
\label{lem:smallSeparatorPlane}
\end{lemma}

The above lemma tells us that a single {\trussnice} truss with large
aspect ratio has a good plane separator.
It turns out that even if we have many such trusses whose longest
dimension may point in different directions, and we have hollowed-out
trusses from small aspect ratio parts, we can still find a single plane that
acts as a reasonably good separator for all of these trusses at the same time.
This is captured by the following lemma,
which is obtained by instantiating Lemma~\ref{lem:kChunkNDMoreDetailed}
with $c_r = 1/3$ and $l = n^{1/6}$.

\begin{lemma}[Combining Separators]
Given a {\esimple} 3D truss
$\calT = \langle V, \{ \pp_i \}_{i \in V}, T, E, \gamma \rangle$,
which is a union of $k$ {\trussnice} trusses with up to $n$ vertices in total. 
Let $\calT' \subset \calT$ be a truss by selectively computing $(B_i, r_i)$-hollowings
of some of the pieces with parameter
\[
r_i \leq n_i^{1/3}.
\]
There exists a randomized algorithm which with high probability returns  a vertex ordering
so that a complete elimination of $\calT'$ has size
$O(n^{23/18}k^{44/9})$,
and takes time $O(n^{11/6}k^{22/3})$ to compute.
\label{lem:kChunkND}
\end{lemma}

The algorithm that achieves Lemma~\ref{lem:kChunkND}
is Algorithm~\ref{alg:NDkChunks} \textsc{ConvexTrussUnionND}
in Section~\ref{sec:NestedDissection}.
Given these lemmas, we can now sketch a proof of Theorem~\ref{thm:main}.

\begin{proof}[Proof of Theorem~\ref{thm:main}]
We bound the running time of
Algorithm~\ref{alg:TrussSolver} $\textsc{TrussSolver}$
with the preconditioner and nested dissection constructed
as per Lemma~\ref{lem:kChunkND}.

Since all the hollowings involve pieces with $r_i \le n_i^{1/3}$,
Definition~\ref{def:hollow} gives a bound of
$O(n^{2/3})$ on the condition
number, and in turn a bound of 
$O(n^{1/3} \log(1 / \epsilon))$ on
the number of PCG iterations via Theorem~\ref{thm:pcg}.
Furthermore, similar to the proof of Theorem~\ref{thm:mainSmallAR},
the Schur complement of $\calT$ onto the elements of $\calT'$
has size $O(n^{10/9})$, 
and computing them by eliminating all interior vertices of our hollowings takes time
$O(n^{4/3})$.

Thus, the total running time of Algorithm~\ref{alg:TrussSolver} is
\[
O\left( n^{4/3}
+ n^{11/6}k^{22/3}
+ n^{1/3} \log\left(1 / \epsilon\right)
\cdot
\left(
n^{10/9}
+
n^{23/18}k^{44/9}
\right)
\right) 
=
O\left(n^{11/6} k^{22/3} \log \left(1/\eps\right)\right).
\]
\end{proof}

\newcommand{\ns}{n}
\newcommand{\tetrahedron}{t}
\newcommand{\gmax}{\phi_{\max}}
\newcommand{\gmin}{\phi_{\min}}
\newcommand{\Trec}{T_{\text{rec}}}
\newcommand{\Tbfs}{T_{\text{BFS}}}
\newcommand{\surface}{\calT}
\newcommand{\dualSurface}{H_{\calS}}
\newcommand{\pa}{\text{pa}}
\newcommand{\tmin}{t_{\min}}
\newcommand{\pt}{\phi_t}

\newcommand\ihat{{\hat{{i}}}}
\newcommand\jhat{{\hat{{j}}}}
\newcommand\khat{{\hat{{k}}}}
\newcommand\shat{{\hat{{s}}}}

\newcommand{\vvtil}{\tilde{\vv}}
\newcommand{\uutil}{\tilde{\uu}}
\newcommand{\vvhat}{\hat{\vv}}
\newcommand{\uuhat}{\hat{\uu}}

\newcommand\ppbar{\boldsymbol{\overline{\mathit{p}}}}
\newcommand\pphat{\boldsymbol{\widehat{\mathit{p}}}}
\newcommand\qqbar{\boldsymbol{\overline{\mathit{q}}}}
\newcommand\qqhat{\boldsymbol{\widehat{\mathit{q}}}}

\newcommand\qNormalized[1]{\qqbar^{\left<#1\right>}}

\section{Bounding the Smallest Nonzero Eigenvalue of a {\esimple} Truss}
\label{sec:minEig}

In this section, we prove Lemma~\ref{lem:3dMinEig}, which lower bounds the smallest nonzero eigenvalue of a {\esimple} 3D truss.
We restate Lemma~\ref{lem:3dMinEig} in the following.

\begin{replemma}{lem:3dMinEig}
Let $\calT$ be an {\esimple} and stiffly-connected 3D truss.
Let $n$ be the number of vertices of $\calT$ and $\Delta$ be the diameter.
Let $\MM$ denote the associated stiffness matrix. 
Then, $\lambda_{\min}(\MM) = \Omega(n^{-1} \Delta^{-4})$ and $\rank(\MM) = 3n - 6$.
\end{replemma}

\subsection{Main Ideas}
\label{sec:mainIdea}

The proof of Lemma~\ref{lem:3dMinEig} is an extension of the path support lemma
by Daitch and Spielman~\cite{daitchS07}.
That proof relies on recentering a vector $\qq$, which is a unit vector orthogonal to the null space,
 with
respect to the a single face by transforming it
along the null space of $\MM$.

Let $\calT = \langle V, \{ \pp_i \}_{i \in V}, T, E, \gamma \rangle$ be a 3D stiffly-connected truss over $n$ vertices.
The null space of the stiffness matrix of $\calT$ 
can be characterized as:
\begin{enumerate}
\item $\pp^{x}$, $\pp^{y}$, $\pp^{z} \in \mathbb{R}^{3n}$:
for each $1 \le i \le n$, the corresponding 3-dimensional vector $\pp_i^x$
($\pp_i^y$ and $\pp_i^z$) has 1 for its $x$-coordinate
($y$-coordinate, and $z$-coordinate, respectively)
and 0 for the other two coordinates.
\item
$\pp^{\perp xy}$, $\pp^{\perp xz}$,
$\pp^{\perp yz} \in \mathbb{R}^{3n}$:
fix an arbitrary index $1\le c \le n$,
for each $1 \le i \le n$:
\begin{align*}
\pp^{\perp xy}_i
&= \left[-\left(\pp_i-\pp_c\right)_{y}, \left(\pp_i-\pp_c\right)_{x}, 0\right]^\top,\\
\pp^{\perp xz}_i
&= \left[-\left(\pp_i-\pp_c\right)_{z}, 0, \left(\pp_i-\pp_c\right)_{x}\right]^\top,\\
\pp^{\perp yz}_i
&= \left[0, -\left(\pp_i-\pp_c\right)_{z}, \left(\pp_i-\pp_c\right)_{y}\right]^\top.
\end{align*}
\end{enumerate}

Also, as many of our arguments are symmetric across dimensions,
we will use $d$, $d_1$ and $d_2$ to represent symmetric indexing
over the dimensions, or pairs of dimensions respectively.
Finally, as centering and exploring a simplicial complex
from a particular triangle introduces an ordering on the
tetrahedrons, faces, and edges,
we will define our edges, triangles, and tetrahedrons as ordered tuples:
\begin{enumerate}
\item Edges: $e = \langle e_1, e_2 \rangle$,
\item Triangles: we denote these as 
$s = \langle s_1, s_2, s_3 \rangle$.
Here $e(s)$ means the edge $\langle s_1, s_2 \rangle$.
\item Tetrahedrons: an ordered 4-tuples of pairwise adjacent points,
$t = \langle t_1, t_2, t_3, t_4 \rangle$.
Here $s(t)$ means the (triangle) surface $\langle t_1, t_2, t_3 \rangle$,
and $e(t)$ means the edge $e(s) = \langle t_1, t_2 \rangle$.
\end{enumerate}

With these notations in mind, 
we can center a vector $\qq$ w.r.t.
a particular
(oriented) triangle surface $s$.
This can be viewed as an extension of the centering lemma in~\cite{daitchS07}.
We prove the following Lemma in Section~\ref{sec:centerVector}.

\begin{lemma}
Given a stiffly-connected, {\esimple} truss $\calT = \langle V, \{ \pp_i \}_{i \in V}, T, E, \gamma \rangle$,
let $\qq$ be a vector orthogonal to the null space of its stiffness matrix.
For each oriented triangle $s$, 
there exists a (unique)  
vector $\qqbar^{\langle s \rangle}$ with scalar shift parameters
$c^{\left<s\right> \perp xy}$,
$c^{\left<s\right> \perp xz}$,
$c^{\left<s\right> \perp yz}$:
\begin{align}
\qqbar^{\langle s \rangle}
=
\qq + \sum_{d1d2} c^{\langle s \rangle \perp d1d2} \pp^{\perp d1d2}
\label{eqn:defC}
\end{align}
satisfying:
\begin{enumerate}
\item the plane containing the points
$\qNormalized{s}_{s}$ is parallel to the plane containing $\pp_{s}$.
\item The edge $\qNormalized{s}_{e(s)}$
is parallel to the edge $\pp_{e(s)}$.
\end{enumerate}
\label{lem:centering}
\end{lemma}

Daitch and Spielman then showed that with any centering,
the value of $\qq^T \MM \qq$ is lower bounded by the sum of a
series of shifted values, or in simpler terms, the norm of
$\qNormalized{s}_{i} -\qNormalized{s}_{j}$ for some edge $ij$.
However, our extension of this bound (which we will
describe next in Lemma~\ref{lem:DS07Main}) to the 3-D case has
an exponential dependency on the distance in tetrahedrons
between $ij$ and $s$.
As a result, we first show the existence of a good centering,
namely one where there exist an edge close to $s$ whose endpoints
are far apart.
This notion of distance can be defined in terms of `hop count'
of tetrahedrons.
\begin{definition}
\label{def:TetDistance}
The \emph{tetrahedron-distance} between a pair of objects $x$ and $y$
in a simplicial complex is the shortest sequence of tetrahedrons
\[
t^{\left(0\right)}, 
t^{\left(1\right)}, 
\ldots
t^{\left(d\right)} 
\]
such that $x \subseteq t^{(0)}$, $y \subseteq t^{(d)}$,
and for all $1 \leq i \leq d$, $t^{(i - 1)}$ and $t^{(i)}$
share a triangle face.
\end{definition}
We remark that because all edges and angles are within some constant
range, this combinatorial distance is within constant factors of
the Euclidean distance of the associated points.
However, we will not make use of this connection.
\begin{lemma}
\label{lem:GoodInitMain}
Given a stiffly-connected, {\esimple} truss $\calT$
with $n$ vertices and diameter $\Delta$,
let $\qq$ be a unit vector orthogonal to the null space of the stiffness matrix of $\calT$.
There exists an oriented triangle $s$
and a pair of points $i,j$ within
tetrahedron-distance $O(1)$ of $s$ satisfying:
\[
\norm{\qNormalized{s}_{i}
-\qNormalized{s}_{j}}_2^2 
= \Omega \left( \frac{1}{ \Delta^4 n} \right).
\]
\end{lemma}
We prove this lemma in Section~\ref{subsec:ExistenceOfGoodCentering}.

We can then check, via an argument similar to~\cite{daitchS07}, that such a centering and distance pair implies a large quadratic form.
The following lemma will be proved in Section~\ref{sec:proofDS07Main}.
\begin{lemma}
\label{lem:DS07Main}
Given a stiffly-connected, {\esimple} truss with stiffness matrix $\MM$, an oriented triangle $s$, and a pair of points $i,j$ within tetrahedron-distance $h$ of $s$, we have
\[
\norm{\qNormalized{s}}_{\MM}^2
\geq 2^{- \Theta \left( h \right)}
\norm{\qNormalized{s}_{i}
-\qNormalized{s}_{j}}_2^2.
\]
\end{lemma}

\begin{proof}[Proof of Lemma~\ref{lem:3dMinEig}]
Consider an arbitrarily fixed unit vector $\qq$ that's orthogonal to the null space of $\MM$.
Let $s$ be the centering given by Lemma~\ref{lem:GoodInitMain},
and let $i,j$ be a pair of points within a constant tetrahedron-distance of $s$.
Lemma~\ref{lem:DS07Main} gives
\begin{align}
\norm{\qNormalized{s}}_{\MM}^2
\geq \Omega(1)
\norm{\qNormalized{s}_{i}
-\qNormalized{s}_{j}}_2^2
\geq \Omega\left( \frac{1}{\Delta^4n} \right).
\end{align}
On the other hand, by Equation~\eqref{eqn:defC}, $\norm{\qNormalized{s}}_2 \ge \norm{\qq}_2 = 1$.
Thus, $\lambda_{\min}(\MM) \ge \Omega\left( \frac{1}{\Delta^4n} \right)$.
\end{proof}

\subsection{Centering a Vector (Proof of Lemma~\ref{lem:centering})}
\label{sec:centerVector}

To prove Lemma~\ref{lem:centering}, 
we define the following operation.
Let $P$ be any fixed plane in $\mathbb{R}^3$ and let $\yy \in \mathbb{R}^3$. 
We define $\yy^{\perp_P}$ to be the vector obtained by first projecting $\yy$ onto plane $P$ and then rotating the projected vector on the plane counterclockwise by $\pi /2$.
The following claim shows that $\yy^{\perp_P}$ can be written as a linear combination of $\yy^{\perp_{xy}}, \yy^{\perp_{yz}}, \yy^{\perp_{xz}}$.

\begin{claim}[Rotation matrix]
Let $P$ be any fixed plane in $\mathbb{R}^3$, and let $\ww$ be its normal vector. 
Then,
\[
\forall \yy \in \mathbb{R}^3, 
\quad 
\yy^{\perp_P} = \ww_z \yy^{\perp_{xy}}
- \ww_y \yy^{\perp_{xz}}  
+ \ww_x \yy^{\perp_{yz}}.
\]
\label{clm:rotation}
\end{claim}

\begin{proof}
Let 
\[
\RR \defeq \left( \begin{array}{ccc}
0 & - \ww_z & \ww_y \\
\ww_z & 0 & - \ww_x \\
- \ww_y & \ww_x & 0
\end{array} \right)
\label{eqn:rotateMatrix}
\]
be the rotation matrix which rotates a vector on the plane $P$ counterclockwise by $\pi / 2$. 
Then,
\[
\yy^{\perp_P} = \RR \left( \yy - (\yy^\top \ww) \ww \right)
= \RR \left(\II - \ww^\top \ww \right) \yy.
\]
We can check that $\RR \ww = {\bf 0}$.
It implies that
$\yy^{\perp_P} = \RR \yy$.

On the other hand,
\[
\ww_z \yy^{\perp_{xy}}
- \ww_y \yy^{\perp_{xz}}  
+ \ww_x \yy^{\perp_{yz}}
= \left( \begin{array}{ccc}
0 & - \ww_z & \ww_y \\
\ww_z & 0 & - \ww_x \\
-\ww_y & \ww_x & 0 
\end{array} \right) \yy
= \RR \yy.
\]
Thus, the claim holds.
\end{proof}

\begin{claim}
Let $\hh \in \mathbb{R}^3$ be a nonzero vector.
Then matrix
\[
\HH \defeq \left( \begin{array}{ccc}
0 & - \hh_z & \hh_y \\
 \hh_z & 0 & - \hh_x \\
- \hh_y & \hh_x & 0 \\
\hh_x & \hh_y & \hh_z
\end{array} \right).
\]
has rank 3.
\label{clm:matrixRank}
\end{claim}
\begin{proof}
The 2-by-2 bottom left submatrix has determinant $-(\hh_y^2 + \hh_x^2)$. If $\hh_y^2 + \hh_x^2 = 0$, then clearly $\HH$ has rank 3 and we have done; otherwise, the 3rd and the 4th rows are independent.

Now it suffices to show that the 2nd row is independent of the 3rd and the 4th rows.
Assume by contradiction, suppose
\begin{align*}
\hh_z &= -\alpha \hh_y + \beta \hh_x, \\
 0 &= \alpha \hh_x + \beta \hh_y, \\
 -\hh_x &= \beta \hh_z.
\end{align*}
By solving the last two equations, we get $\beta = -\hh_x / \hh_z$ and $\alpha = \hh_y / \hh_z$. 
Plugging these values into the 1st equation, we have
$\hh_x^2 + \hh_y^2 + \hh_z^2 = 0$, which contradicts that $\hh \neq {\bf 0}$.
\end{proof}

\begin{proof}[Proof of Lemma~\ref{lem:centering}]
Let $\ww$ be the normal vector of the plane containing $s = \langle \pp_{i_1}, \pp_{i_2}, \pp_{i_3} \rangle$, that is,
\[
\ww^\top (\pp_{i_2} - \pp_{i_1}) = 0, 
\ww^\top (\pp_{i_3} - \pp_{i_1}) = 0,
\mbox{ and }
\norm{\ww}_2 = 1.
\]

We first show that there exist $\alpha_{xy}, \alpha_{xz}, \alpha_{yz} \in \mathbb{R}$ such that the vector
\[
\zz \defeq
\qq + \sum_{d1d2} \alpha_{d1d2} \pp^{\perp_{d1d2}}
\]
satisfies the first condition.
It suffices to show that the following linear system has a solution for real numbers $\alpha_{d1d2}$'s:
\begin{align*}
\ww^\top (\qq_{i_2} - \qq_{i_1} +\sum_{d1d2} \alpha_{d1d2} (\pp_{i_2} - \pp_{i_1})^{\perp_{d1d2}}) &= 0, \\
\ww^\top (\qq_{i_3} - \qq_{i_1} +\sum_{d1d2} \alpha_{d1d2} (\pp_{i_3} - \pp_{i_1})^{\perp_{d1d2}}) &= 0. 
\end{align*}
Rearrange it and write it in matrix form:
\[
\left( \begin{array}{ccc}
\ww^\top (\pp_{i_2} - \pp_{i_1})^{\perp_{xy}} & \ww^\top (\pp_{i_2} - \pp_{i_1})^{\perp_{yz}} & \ww^\top (\pp_{i_2} - \pp_{i_1})^{\perp_{xz}} \\
\ww^\top  (\pp_{i_3} - \pp_{i_1})^{\perp_{xy}} & \ww^\top (\pp_{i_3} - \pp_{i_1})^{\perp_{yz}} & \ww^\top (\pp_{i_3} - \pp_{i_1})^{\perp_{xz}} 
\end{array} \right) 
\left( \begin{array}{c}
\alpha_{xy} \\
\alpha_{yz} \\
\alpha_{xz}
\end{array} \right)
=
\left( \begin{array}{c}
-\ww^\top(\xx_{i_2} - \xx_{i_1}) \\
-\ww^\top(\xx_{i_3} - \xx_{i_1})
\end{array} \right).
\]
It suffices to show that the 
coefficient matrix has rank 2.

Assume by contradiction, there is some $k \in \mathbb{R}$ such that
\[
\ww^\top (\pp_{i_2} - \pp_{i_1})^{\perp_{d1d2}} = k \ww^\top (\pp_{i_3} - \pp_{i_1})^{\perp_{d1d2}}, 
\quad \forall d1,d2
\]
It equals to
\[
\ww^\top  \left( (\pp_{i_2} - \pp_{i_1}) - k(\pp_{i_3} - \pp_{i_1}) \right)^{\perp_{d1d2}} = 0,
\quad \forall d1,d2
\]
Let $\hh \defeq (\pp_{i_2} - \pp_{i_1}) - k(\pp_{i_3} - \pp_{i_1})$.
Since vectors $\pp_{i_2} - \pp_{i_1}, \pp_{i_3} - \pp_{i_1}$ are not parallel, we have $\hh \neq {\bf 0}$.
Besides, $\hh \perp \ww$.
Write the above equation in matrix form:
\[
\left( \begin{array}{ccc}
0 & - \hh_z & \hh_y \\
 \hh_z & 0 & - \hh_x \\
- \hh_y & \hh_x & 0 \\
\hh_x & \hh_y & \hh_z
\end{array} \right)  \ww
={\bf 0}.
\]
By Claim~\ref{clm:matrixRank}, the coefficient matrix has rank 3, which implies that $\ww = {\bf 0}$.
It
contradicts that $\norm{\ww}_2 = 1$.

Then we show that there exist $\beta_{xy}, \beta_{xz}, \beta_{yz} \in \mathbb{R}$ such that 
\[
\qq^{\langle s \rangle}= \zz + \sum_{d1d2} \beta_{d1d2} \pp^{\perp_{ d1d2}}
\]
satisfies both conditions.

Let $P$ be the plane containing $\pp_{i_1}, \pp_{i_2}, \pp_{i_3}$.
Let $\gg \in \mathbb{R}^{3\ns}$ satisfy
\[
\gg_i = (\pp_i - \pp_1)^{\perp_P}, 
\quad \forall 1 \le i \le \ns
\]
Then, there exists an appropriate multiplier $\gamma \in \mathbb{R}$ such that the vector
\[
(\zz_{i_2} - \zz_{i_1}) + \gamma (\gg_{i_2} - \gg_{i_1})
= (\zz_{i_2} - \zz_{i_1})
+ \gamma (\pp_{i_2} - \pp_{i_1})^{\perp_P}
\]
is parallel to $\pp_{i_2} - \pp_{i_1}$.
Besides, since both $\zz$ and $\gg$ are parallel to the plane $P$, the vector $\qq^{\langle s \rangle} = \zz + \gamma \gg$ is parallel to the plane $P$.

By Claim~\ref{clm:rotation}, there exists real numbers $\beta_{xy}, \beta_{xz}, \beta_{yz}$ such that 
\[
\gamma \gg = \sum_{d1d2} \beta_{d1d2} \pp^{\perp_{d1d2}}.
\]
Let 
\[
c^{\langle s \rangle \perp_{d1d2}} = \alpha_{d1d2}+ \beta_{d1d2}, 
\quad \forall d1,d2
\]
This completes the proof.
\end{proof}

\subsection{Existence of Good Centering (Proof of Lemma~\ref{lem:GoodInitMain})}
\label{subsec:ExistenceOfGoodCentering}

In this section we prove Lemma~\ref{lem:GoodInitMain}.

The proof is by contradiction.
We assume that for \emph{every} centering
at a triangle $s$,
for every pair $\langle i, j \rangle$ within tetrahedron-distance 
3 of $\left<s\right>$ satisfies
\begin{align}
\norm{\qNormalized{s}_{i} -\qNormalized{s}_{j}}_2
\leq \epsilon,
\label{eqn:assumption}
\end{align}
where $\epsilon \defeq \sqrt{\frac{1}{ \beta  \Delta^4 n}}$ for a sufficiently large constant $\beta$.

Recall that in Lemma~\ref{lem:centering},
for each centering at $\left<s\right>$, we define 3 scalar coefficients
\[
c^{\left<s\right> \perp xy},
c^{\left<s\right> \perp xz},
c^{\left<s\right> \perp yz} \in \mathbb{R}
\]
for the $3$ null space vectors in Equation~\eqref{eqn:defC}.
We will write the vector containing these 3 coefficients as $\cc^{\left<s\right>}$.

To prove Lemma~\ref{lem:GoodInitMain}, we need the following lemma.
It says that, under the assumption in Equation~\eqref{eqn:assumption}, the difference between the coefficient vectors w.r.t. to two close centering triangles is small.

\begin{lemma}
\label{lem:CoefficientsClose}
Assume that for every centering triangle $s$ and every pair of points $i, j$ within
distance $3$ of $\left<s\right>$ satisfies Equation~\eqref{eqn:assumption}. 
Then for every pair of centering triangles $s_1$ and $s_2$ within constant tetrahedron distance to each other,  
we have
\[
\norm{\cc^{\left<s_1\right>} - \cc^{\left<s_2\right>}}_2
= O (\eps).
\]
\end{lemma}
\newcommand{\jsum}{\sum_{d_1d_2}}

The next lemma implies that: for any two vertices such that each is centered w.r.t. a triangle close to itself, the difference between the two centered vertices is small.

\begin{lemma}
Let $u, w$ be two arbitrary vertices of $\calT$. 
Let $s_u$ ($s_w$) be a triangle containing $u$ (and $w$, respectively).
Under the assumption in Equation~\eqref{eqn:assumption}, we have
\[
\norm{\qNormalized{s_u}_{u} - \qNormalized{s_w}_{w}}_2 = O\left( \Delta^2 \epsilon \right).
\]
\label{lem:diffTwoVertices}
\end{lemma}

\begin{proof}
Let $u = v_1, v_2, \ldots, v_f = w$ be a shortest path from $u$ to $w$.
Note $f \le \Delta$.
Let $s_i$ be a triangle next to $v_i$ for $2 \le i \le f-1$.
The path from vertex $u$, centered at $\langle s_u \rangle$, to vertex $w$, centered at $\langle s_w \rangle$, can be expressed as the following:
\[
\qNormalized{s_u}_{u} - \qNormalized{s_w}_{w} = 
\qNormalized{s_u}_{u} - \qNormalized{s_u}_{v_2} 
+ \sum_{2\le i \le f-1} \left( \qNormalized{s_{i}}_{v_{i}} -\qNormalized{s_{i+1}}_{v_{i}} \right)
+ \sum_{2 \le i \le f-1} \left( \qNormalized{s_{i+1}}_{v_{i}} -\qNormalized{s_{i+1}}_{v_{i+1}} \right).
\]
Taking $\ell_2$ norm on both sides and applying the triangle inequality, we an bound the norm of the LHS by the sum of $\ell_2$ norm of each term in the RHS.

By Equation~\eqref{eqn:defC} and the triangle inequality, 
\[
\norm{\qNormalized{s_{i}}_{v_i} - \qNormalized{s_{i+1}}_{v_{i}}}_2 
\le \jsum{ \abs{\cc^{\left<s_i\right> \perp d1d2} - \cc^{\left<s_{i+1}\right> \perp d1d2}}
\norm{(\pp_{v_i} - \pp_c)^{\perp d1d2}}}_2.
\]
Apply Lemma~\ref{lem:CoefficientsClose}:
\begin{align}
\norm{\qNormalized{s_{i}}_{v_i} - \qNormalized{s_{i+1}}_{v_{i}}}_2
= O(\Delta \eps).
\label{eqn:diffTwoTriangles}
\end{align}
Together with our assumption in Equation~\eqref{eqn:assumption}, we have
\[
\norm{\qNormalized{s_u}_{u} - \qNormalized{s_w}_{w}}_2
\le O(\Delta^2 \eps).
\]
\end{proof}

Now we prove Lemma~\ref{lem:GoodInitMain}.

\begin{proof}[Proof of Lemma~\ref{lem:GoodInitMain}]
For each vertex $i$, let $s_i$ denote an arbitrary triangle next to vertex $i$.
We can write vector $\qq$ as
\begin{align}
\qq = \qqhat + 
\qqtilde - \sum_{d_1d_2} 
c^{\left<s_1 \right>\perp d_1d_2} \pp^{\perp d_1 d_2} + \ee,
\label{eqn:qRewrite}
\end{align}
where 
\[
\qqhat_i = \qNormalized{s_i}_i - \qNormalized{s_1}_1, \qqtilde_i = \qNormalized{s_1}_i - \qNormalized{s_i}_i,\ee_i = \qNormalized{s_1}_1.
\]
Note that the last two terms of Equation~\eqref{eqn:qRewrite} are in the null space of $\MM$.
Thus, 
\[
\norm{\qqhat + \qqtilde}_2 \ge \norm{\qq}_2 = 1.
\]
On the other hand, by the triangle inequality, 
\[
\norm{\qqhat + \qqtilde}_2  \le \norm{\qqhat}_2 + \norm{\qqtilde}_2
 = \left( \sum_i \norm{\qNormalized{s_i}_i - \qNormalized{s_1}_1}_2^2 \right)^{1/2}
+ \left( \sum_i \norm{\qNormalized{s_1}_i - \qNormalized{s_i}_i}_2^2 \right)^{1/2}.
\]
We apply Lemma~\ref{lem:diffTwoVertices} to the first term, and apply Equation~\eqref{eqn:diffTwoTriangles} (which is true for any two close centering triangles) $\Delta$ times for the second term:
\[
\norm{\qqhat + \qqtilde}_2  \le \norm{\qqhat}_2 + \norm{\qqtilde}_2
= O(\sqrt{n} \Delta^2 \eps).
\]
By our choice of $\eps \leftarrow \frac{1}{  \beta  \sqrt{n} \Delta^2 }$
for a sufficiently large constant $\beta$, we get 
a contradiction.
\end{proof}

It remains to prove Lemma~\ref{lem:CoefficientsClose}.
For it, we define a matrix w.r.t. a given vector $\vv \in \mathbb{R}^3$:
\begin{align}
\QQ_{\vv}
\defeq \left( \vv^{\perp{yz}}, \vv^{\perp{xz}}, \vv^{\perp{xy}} \right)
= \left( \begin{array}{ccc}
0 & -\vv_z & \vv_y \\
\vv_z & 0 & - \vv_x \\
-\vv_y & \vv_x & 0
\end{array} \right).
\label{eqn:Q}
\end{align}
We will use the following properties of $\QQ_{\vv}$.
\begin{lemma}
\label{lem:Qeig}
For the matrix $\QQ_{\vv}$ as defined in Equation~\eqref{eqn:Q},
its singular values are $0, 1, 1$, and its null space is
multiples of the vector $\vv$.
\end{lemma}

\begin{proof}
For any vector $\pp \in \mathbb{R}^3$, the cross product $\vv \times \pp = \QQ_{\vv} \pp$.
That is, $\QQ_{\vv} \pp$ is the vector obtained by:
\begin{enumerate}
\item First projecting $\pp$ onto the plane with normal vector $\vv$,
say	$P_{\perp \vv}$, and
\item then rotating the projected vector on the plane $P_{\perp \vv}$
by $\frac{\pi}{2}$ counterclockwise.
\end{enumerate}
From this description, we can infer that
\[
\nulls(\QQ_{\vv}) = \Span{\vv}
\]
and any vector orthogonal to $\nulls(\QQ_{\vv})$ is on this plane $P_{\perp \vv}$.
Such vectors are not affected by the first step projection,
and their lengths are not changed by the subsequent rotation.
This means for such vectors $\uu \perp \nulls(\QQ_{\vv})$
we have $\uu^\top \QQ_{\vv}^\top \QQ_{\vv} \uu  = \norm{\uu}_2^2$.
Thus, the singular values of $\QQ_{\vv}$ are 1,1,0.
\end{proof}

\begin{lemma}
\label{lem:QTwo}
If $\vv$ and $\uu$ are vectors with length at least $1$ such that
the angle between them is 
$\theta \in (0, \pi)$, then the matrix
\[
\QQ_{\vv}^{\top}\QQ_{\vv}
+ \QQ_{\uu}^{\top}\QQ_{\uu}
\]
is full rank, and has minimum eigenvalue at least $2 \sin^2 \frac{\theta}{2}$.
\end{lemma}

\begin{proof}

Let $\hh \in \mathbb{R}^3$ be a unit vector.
Decompose $\hh$:
\[
\hh  = \alpha_v \vv + \beta_v \vvhat = \alpha_u \uu + \beta_u \uuhat,
\]
where $\alpha_v, \alpha_u, \beta_v, \beta_u \in \mathbb{R}$, $\vvhat$ is a unit vector orthogonal to $\vv$, and $\uuhat$ is a unit vector orthogonal to $\uu$.
Then,
\[
\hh^\top (\QQ_v^\top \QQ_v + \QQ_u^\top \QQ_u) \hh
= \beta_v^2 \norm{\QQ_v \vvhat}_2^2 + \beta_u^2 \norm{\QQ_u \uuhat}_2^2.
\]
By Lemma~\ref{lem:Qeig}, $\norm{\QQ_v \vvhat}_2 = 1$ and $\norm{\QQ_u \uuhat}_2 = 1$. Thus,
\[
\hh^\top (\QQ_v^\top \QQ_v + \QQ_u^\top \QQ_u) \hh = \beta_v^2 + \beta_u^2
= 2  - (\alpha_v^2 + \alpha_u^2).
\]
The second equality is due to that $\hh$ is a unit vector.
\[
\alpha_v^2 + \alpha_u^2 = \norm{\hh^\top \vv}_2^2
+ \norm{\hh^\top \uu}_2^2
\le 2 \cos^2 \frac{\theta}{2}.
\]
Therefore,
\[
\lambda_{\min} (\QQ_v^\top \QQ_v + \QQ_u^\top \QQ_u) \ge 2 \sin^2 \frac{\theta}{2}.
\]
$\theta < \pi$ implies that $\lambda_{\min} (\QQ_v^\top \QQ_v + \QQ_u^\top \QQ_u) > 0$, and thus the matrix $\QQ_v^\top \QQ_v + \QQ_u^\top \QQ_u$ has full rank.
\end{proof}

Equipped with the above lemmas, we prove Lemma~\ref{lem:CoefficientsClose}.
\begin{proof}[Proof of Lemma~\ref{lem:CoefficientsClose}]

Let $s_1$ and $s_2$ be two triangle centerings 
for which there exist two edges belong to same tetrahedron, say $(i,j), (j,k)$,
such that vertices $i,j,k$ are all within 
tetrahedron-distance constant $h$ of both $s_1$ and $s_2$.

Recall $\qNormalized{s_1}$ is defined in Equation~\eqref{eqn:defC}.
Subtracting
\begin{align*}
\qNormalized{s_1}_{i} 
& = \qq_{i} + \sum_{d_1 d_2} c^{\left<s_1 \right> \perp d_1 d_2}
  \left( \pp_{i}^{\perp d_1 d_2} - \pp_{c}^{\perp d_1 d_2} \right)
\end{align*}
from the corresponding equation for $j$ gives:
\begin{align*}
\qNormalized{s_1}_{i} - \qNormalized{s_1}_{j}
& = \qq_{i} - \qq_{j} + \sum_{d_1 d_2} c^{\left<s_1 \right> \perp d_1 d_2}
  \left( \pp_{i}^{\perp d_1 d_2} - \pp_{j}^{\perp d_1 d_2} \right).
\end{align*}
Plugging in the definition of $\QQ_{(\pp_{i} - \pp_{j})}$:
\begin{align*}
\qNormalized{s_1}_{i} - \qNormalized{s_1}_{j}
& = \qq_{i} - \qq_{j} +
	\QQ_{(\pp_{i} - \pp_{j})} \cc^{\left< s_1 \right>}.
\end{align*}
Subtracting this equation from centering triangle $s_2$
in turn
cancels the $\qq_{i} - \qq_{j}$ term on the RHS, giving:
\[
\QQ_{(\pp_{i} - \pp_{j})}
\left(\cc^{\left< s_1 \right>} - \cc^{\left< s_2 \right>}\right)
 = 
\left( \qNormalized{s_1}_{i} - \qNormalized{s_1}_{j} \right)
-
\left( \qNormalized{s_2}_{i} - \qNormalized{s_2}_{j} \right).
\label{eqn:CtoQ}
\]
Together with the corresponding equation for $j,k$,
\[
\left( \begin{array}{c}
\QQ_{(\pp_{i} - \pp_{j})} \\
\QQ_{(\pp_{j} - \pp_{k})}
\end{array} \right) 
\left(\cc^{\left< s_1 \right>} - \cc^{\left< s_2 \right>}\right)
= \left( \begin{array}{c}
\left( \qNormalized{s_1}_{i} - \qNormalized{s_1}_{j} \right)
-
\left( \qNormalized{s_2}_{i} - \qNormalized{s_2}_{j} \right) \\
\left( \qNormalized{s_1}_{j} - \qNormalized{s_1}_{k} \right)
-
\left( \qNormalized{s_2}_{j} - \qNormalized{s_2}_{k} \right)
\end{array} \right).
\]
Multiplying $\left( \begin{array}{cc}
\QQ_{(\pp_{i} - \pp_{j})}^\top &
\QQ_{(\pp_{j} - \pp_{k})}^\top
\end{array}\right)$ on both sides gives:
\begin{multline*}
\left( \QQ_{(\pp_{i} - \pp_{j})}^\top \QQ_{(\pp_{i} - \pp_{j})}
+ \QQ_{(\pp_{j} - \pp_{k})}^\top \QQ_{(\pp_{j} - \pp_{k})} \right) \left(\cc^{\left< s_1 \right>} - \cc^{\left< s_2 \right>}\right) \\
=   \left( \begin{array}{cc}
\QQ_{(\pp_{i} - \pp_{j})}^\top &
\QQ_{(\pp_{j} - \pp_{k})}^\top
\end{array}\right) 
\left( \begin{array}{c}
\left( \qNormalized{s_1}_{i} - \qNormalized{s_1}_{j} \right)
-
\left( \qNormalized{s_2}_{i} - \qNormalized{s_2}_{j} \right) \\
\left( \qNormalized{s_1}_{j} - \qNormalized{s_1}_{k} \right)
-
\left( \qNormalized{s_2}_{j} - \qNormalized{s_2}_{k} \right)
\end{array} \right).
\end{multline*}
Solving the above linear equations and taking norm on both sides:
\begin{multline*}
\norm{\cc^{\left< s_1 \right>} - \cc^{\left< s_2 \right>}}_2
\le  \norm{\left( \QQ_{(\pp_{i} - \pp_{j})}^\top \QQ_{(\pp_{i} - \pp_{j})}
+ \QQ_{(\pp_{j} - \pp_{k})}^\top \QQ_{(\pp_{j} - \pp_{k})} \right)^{-1}}_2 \\
 \cdot \norm{\left( \begin{array}{cc}
\QQ_{(\pp_{i} - \pp_{j})}^\top &
\QQ_{(\pp_{j} - \pp_{k})}^\top
\end{array}\right) }_2
\norm{\left( \begin{array}{c}
\left( \qNormalized{s_1}_{i} - \qNormalized{s_1}_{j} \right)
-
\left( \qNormalized{s_2}_{i} - \qNormalized{s_2}_{j} \right) \\
\left( \qNormalized{s_1}_{j} - \qNormalized{s_1}_{k} \right)
-
\left( \qNormalized{s_2}_{j} - \qNormalized{s_2}_{k} \right)
\end{array} \right)}_2.
\end{multline*}
Applying Lemma~\ref{lem:Qeig} and~\ref{lem:QTwo} on the first two terms, 
 and applying the triangle inequality and the assumption in Equation~\eqref{eqn:assumption} give:
 \[
\norm{\cc^{\left< s_1 \right>} - \cc^{\left< s_2 \right>}}_2
\le O( \eps). 
 \]
\end{proof}

\section{Proof of Path Lemma}
\label{sec:proofDS07Main}

We now prove Lemma~\ref{lem:DS07Main}, which lower bounds the quadratic form of an {\esimple} and stiffly-connected truss stiffness matrix 
by the distance between two centered points.
As we now only deal with a single centering in this section, we will drop this superscription for simplicity and relabel all indices w.r.t. this centering.
Our relabeling is similar to that in~\cite{daitchS07}.

\subsection{Relabeling Tetrahedrons and Vertices}

Fix an arbitrary tetrahedron, say $\tetrahedron_1$, and we center the vector $\qq$ w.r.t. one of the triangle faces of $\tetrahedron_1$ as in Lemma~\ref{lem:centering}.
Use $\tetrahedron_1$ as root, we
run breadth-first-search (BFS) in the rigidity graph (refer to Definition~\ref{def:stifflyConnected}) of $\surface$, and relabel the tetrahedrons of $\surface$ according to this BFS ordering. 
For example, the neighbor tetrahedrons of $\tetrahedron_1$ are labeled as $\tetrahedron_2, \tetrahedron_3, \ldots$.
Let $\Tbfs$ be the corresponding BFS tree in which each node represents a tetrahedron in $\surface$.

Based on $\Tbfs$,
we relabel the vertices of $\surface$ as follows.
We label the vertices of $\tetrahedron_1$ by $-2,-1,0,1$ in an arbitrary order.  Each child of $\tetrahedron_1$ (and subsequent recursions)  shares a triangle face with their respective parent tetrahedron and thus only requires us to label one vertex per child tetrahedron, which can be labeled according to the BFS ordering.

For each newly labeled vertex $j$, we use $\tetrahedron_j$ to denote the tetrahedron encompassing the $\tetrahedron_j$ and its parental face.
Besides,
we use a 3-dimensional vector $\sigma_j$ to consist of 
the indexes of the parental face sorted in ascending order, and is assumed that indexing this vector is implicitly modulo 3.
Then, $t_j = \{j, \sigma_j(1), \sigma_j(2), \sigma_j(3)\}$.
For completeness, we define $\sigma_1 \defeq \{-2,-1,0\}$. See Figure~\ref{fig:relabel}  for an example. 

\begin{figure}[htb]
\centering
\includegraphics[scale=0.55]{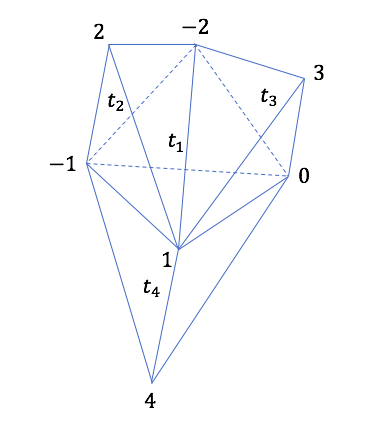}
\caption{An example of the relabeling of tetrahedrons and vertices: $t_1 = \{-2,-1,0,1\}, t_2 = \{-2, -1, 1, 2\}, t_3 = \{-2,0,1,3\}, t_4 = \{-1,0,1,4\}$ and $ \sigma_1 = \{-2,-1,0 \}, \sigma_2 = \{-2,-1,1\}, \sigma_3 = \{-2,1,0\}, \sigma_4 = \{0,-1,1\}$.}
\label{fig:relabel}
\end{figure}

\subsection{Distance between Local Minimizers and the Centered Vectors}

Let $\qq \in \mathbb{R}^{3\ns}$ be a unit vector which is orthogonal to $\Span{\pp^x,\pp^y,\pp^z,\pp^{\perp xy},\pp^{\perp yz},\pp^{\perp xz}}$ defined at the beginning of Section~\ref{sec:mainIdea}.
By Lemma~\ref{lem:centering},
there exist scalars $c^{\left<t_1\right> \perp xy},c^{\left<t_1\right> \perp xz},c^{\left<t_1\right> \perp yz} \in \mathbb{R}$ such that the vector
\[
\qNormalized{t_1} \defeq \qq + \jsum{ c^{\left<t_1\right> \perp d_1d_2} \pp^{\perp d_1d_2}}
\]
satisfies:
\begin{enumerate}
\item the plane containing $\qNormalized{t_1}_{-2}, \qNormalized{t_1}_{-1}, \qNormalized{t_1}_{0}$ is parallel to the plane $\sigma_{\tetrahedron_{1}}$, and
\item $\qNormalized{t_1}_{-2} - \qNormalized{t_1}_{-1}$ is parallel to $(\sigma_{\tetrahedron_1}(1)-\sigma_{\tetrahedron_1}(2))$.
\end{enumerate}

We drop the superscription $\langle t_1 \rangle$ when the context is clear.

For each $0 \le i \le n-3$, we define a 3-dimensional vector $\yy_i$. 
$\yy_0$ is a vector on the plane of $\sigma_{t_1} = \{\pp_{-2}, \pp_{-1}, \pp_{0}\}$ and minimizes the energy / quadratic form, suppose both $\qNormalized{t_1}_{-2}, \qNormalized{t_1}_{-1}$ are fixed.
That is,
\begin{equation}
\label{eqn:x0def}
\begin{split}
&(\pp_{0} - \pp_{j})^\top (\yy_{0} - \qqbar_{j}) = 0, \quad \forall j \in \{-2, -1\} \\
&\ww^\top (\yy_{0} - \qqbar_{-2}) = 0
\end{split}.
\end{equation}
For each $1 \le i \le n-3$, $\yy_i$ is a vector which minimizes the energy / quadratic form, suppose all the three points of $\sigma_j$ are fixed.
That is,
\begin{align}
(\pp_{i} - \pp_{\sigma_{i}(j)} )^\top (\yy_{i} - \qqbar_{\sigma_{i}(j)}) = 0, \quad \forall j \in \{1,2,3\}
\label{eqn:xiDef}
\end{align}

We then define the distance between each local minimizer $\yy_i$ and the centered vector $\qqbar_i$.
Specifically,
\begin{align}
\dd_i \defeq \left\{ \begin{array}{ll}
\qqbar_{-1} - \qqbar_{-2}, & \quad i = -1 \\
\qqbar_i - \yy_i, & \quad 0 \le i \le n-3
\end{array}
\right.
\label{eqn:diDef}
\end{align}
This definition intermediately gives that $\forall -1 \le i \le \ns-3, 1 \le j \le \min\{i+2,3\}$,
\begin{align}
\dd_i^\top (\pp_i - \pp_{\sigma_i(j)}) 
 = (\qqbar_i - \yy_i)^\top (\pp_i - \pp_{\sigma_i(j)})
 = (\qqbar_i - \qqbar_{\sigma_i(j)})^\top (\pp_i - \pp_{\sigma_i(j)}).
\label{eqn:diReplace}
\end{align}
Note that equation (\ref{eqn:diReplace}) describes the net stress on an edge.

We can explicitly express these $\dd_i$'s
by solving Equations~\eqref{eqn:x0def} and~\eqref{eqn:xiDef}.

Let 
\[
\FF_0 \defeq  \left( \begin{array}{c}
(\pp_0 - \pp_{-2})^\top \\
(\pp_0 - \pp_{-1})^\top \\
\ww^\top
\end{array} \right).
\]
Let $\hh_{-2}, \hh_{-1}, \hh_{\ww} \in \mathbb{R}^3$ such that for each $1 \le k \le 3$,
\begin{align*}
\hh_{-2}(k) &\defeq \det \left( \begin{array}{cc}
(\pp_0 - \pp_{-1})^{(k+1)_3} & (\pp_0 - \pp_{-1})^{(k+2)_3} \\
\ww^{(k+1)_3} & \ww^{(k+2)_3}
\end{array} \right), \\
\hh_{-1}(k) &\defeq \det \left( \begin{array}{cc}
(\pp_0 - \pp_{-2})^{(k+1)_3} & (\pp_0 - \pp_{-2})^{(k+2)_3} \\
\ww^{(k+1)_3} & \ww^{(k+2)_3}
\end{array} \right), \\
\hh_{\ww}(k) &\defeq \det \left( \begin{array}{cc}
(\pp_0 - \pp_{-1})^{(k+1)_3} & (\pp_0 - \pp_{-1})^{(k+2)_3} \\
(\pp_0 - \pp_{-2})^{(k+1)_3} & (\pp_0 - \pp_{-2})^{(k+2)_3}
\end{array} \right).
\end{align*}
where $(k+1)_3$ and $(k+2)_3$ denotes accessing specific dimensions of vectors $\in \mathbb{R}^3$ modulo 3.
Then we define 
\[
\HH_{0,2} \defeq \hh_{-2} (\pp_0 - \pp_{-2})^\top,
\quad 
\HH_{0,1} \defeq \hh_{-1} (\pp_0 - \pp_{-1})^\top,
\quad \mbox{and }
\HH_{\ww} \defeq \hh_{\ww} \ww^\top.
\]

Similarly, for $1\le i \le \ns-3$, let
\[
\FF_i \defeq
\left( \begin{array}{c}
(\pp_i - \pp_{\sigma_i(1)})^\top \\
(\pp_i - \pp_{\sigma_i(2)})^\top \\
(\pp_i - \pp_{\sigma_i(3)})^\top 
\end{array} \right),
\]
and 
\begin{align}
\HH_{i,j} \defeq 
\hh_{i,j}
(\pp_i - \pp_{\sigma_i(j)})^\top, \quad \forall 1 \le j \le 3,
\label{eqn:HijDef}
\end{align}
where $\hh_{i,j} \in \mathbb{R}^3$ with
\[
\hh_{i,j}(k) \defeq  \det \left( \begin{array}{cc}
(\pp_i - \pp_{\sigma_i((j+1)_3)})^{(k+1)_3} & (\pp_i - \pp_{\sigma_i((j+1)_3)})^{(k+2)_3} \\
(\pp_i - \pp_{\sigma_i((j+2)_3)})^{(k+1)_3} & (\pp_i - \pp_{\sigma_i((j+2)_3)})^{(k+2)_3}
\end{array} \right) .
\]

We can check that
the following $\yy_i$'s satisfy Equations~\eqref{eqn:x0def} and~\eqref{eqn:xiDef}:
\[
\yy_i = \left\{ \begin{array}{ll}
\frac{1}{\det(\FF_0)} \left( (\HH_{0,2} + \HH_{\ww} ) \qqbar_{-2} + \HH_{0,1} \qqbar_{-1}  \right), & \quad  i = 0 \\
\frac{1}{\det(\FF_i)} \sum_{1\le j\le 3} \HH_{i,j} \qqbar_{\sigma_i(j)},
& \quad \forall 1 \le i \le n
\end{array} \right.
\]

\begin{claim}
\begin{align}
\det(\FF_i) = (\pp_i - \pp_{\sigma_i(j)})^\top \hh_{i,j}.
\label{eqn:detFi}
\end{align}
\label{clm:determinant}
\end{claim}
\begin{proof}
Let 
\[
\QQ = \left( \begin{array}{ccc}
a & b & c \\
d & e & f \\
g & h & i
\end{array} \right)
\]
be a $3 \times 3$ matrix.
Then, the determinant of $\QQ$ is
\[
\det(\QQ)
= a \cdot \det\left( \left( \begin{array}{cc}
e & f \\
h & i
\end{array} \right) \right)
- b \cdot \det\left( \left( \begin{array}{cc}
d & f \\
g & i
\end{array} \right) \right)
+ c \cdot \det\left( \left( \begin{array}{cc}
d & e \\
g & h
\end{array} \right) \right).
\]
Applying the above rule to $\FF_i$ gives Eqaution~\eqref{eqn:detFi}.
\end{proof}

\begin{claim}
\[
\det(\FF_i) \II = \sum_{1 \le j \le 3} \HH_{i,j}, 
\quad \forall i \ge 0
\]
\label{clm:FHrelation}
\end{claim}

\begin{proof}
We first show that the diagonals of $\sum_{j \in [3]} \HH_{i,j}$ are equal to $\det(\FF_i)$.
\begin{align*}
& \sum_{j \in [3]} \HH_{i,j} (1,1) \\
= & \sum_{j \in [3]}  (\pp_i - \pp_{\sigma_i(j)})(1) \cdot \det \left( \begin{array}{cc}
(\pp_i - \pp_{\sigma_i((j+1)_3)})(2) & (\pp_i - \pp_{\sigma_i((j+1)_3)})(3) \\
(\pp_i - \pp_{\sigma_i((j+2)_3)})(2) & (\pp_i - \pp_{\sigma_i((j+2)_3)})(3)
\end{array} \right) \\
=&   \det \left( \begin{array}{c}
(\pp_i - \pp_{\sigma_i(1)})^\top \\
(\pp_i - \pp_{\sigma_i(2)})^\top \\
(\pp_i - \pp_{\sigma_i(3)})^\top 
\end{array} \right) \\
= & \det(\FF_i).
\end{align*}
Similarly, we have
\[
\sum_{j \in [3]} \HH_{i,j} (2,2)  = \sum_{j \in [3]} \HH_{i,j} (3,3)
= \det(\FF_i).
\]

Then, we show that the off-diagonals of $\sum_{j \in [3]} \HH_{i,j}$ are 0.
\begin{align*}
& \sum_{j \in [3]} \HH_{i,j} (1,2) \\
= & \sum_{j \in [3]}  (\pp_i - \pp_{\sigma_i((j+1)_3)})(2) \cdot \det \left( \begin{array}{cc}
(\pp_i - \pp_{\sigma_i((j+1)_3)})(2) & (\pp_i - \pp_{\sigma_i((j+1)_3)})(3) \\
(\pp_i - \pp_{\sigma_i((j+2)_3)})(2) & (\pp_i - \pp_{\sigma_i((j+2)_3)})(3)
\end{array} \right) \\
= & \det  \left( \begin{array}{ccc}
(\pp_i - \pp_{\sigma_i(1)})(2) & (\pp_i - \pp_{\sigma_i(1)})(2) & (\pp_i - \pp_{\sigma_i(1)})(3) \\
(\pp_i - \pp_{\sigma_i(2)})(2) & (\pp_i - \pp_{\sigma_i(2)})(2) & (\pp_i - \pp_{\sigma_i(2)})(3)\\
(\pp_i - \pp_{\sigma_i(3)})(2) & (\pp_i - \pp_{\sigma_i(3)})(2) & (\pp_i - \pp_{\sigma_i(3)})(3) 
\end{array} \right) \\
= & 0.
\end{align*}
The last equation is due to the 3 columns of the matrix are linearly dependent.
Similarly, 
\[
 \sum_{j \in [3]} \HH_{i,j} (s,t) = 0, 
 \quad \forall 1 \le s \neq t \le 3.
\]
This completes the proof.
\end{proof}

Plugging the above equations into the definition of $\dd_i$'s gives:
\begin{claim}
For each $-1 \le i \le \ns-3$,
\[
\dd_i = \qqbar_i - \qqbar_{\sigma_i(j)}
+ \frac{1}{\det(\FF_i)} \sum_{\substack{1 \le j' \le \min\{i+2, 3\} \\ j' \neq j}} \HH_{i,j'} \left( \qqbar_{\sigma_i(j)} - \qqbar_{\sigma_i(j')} \right).
\]
\label{clm:diSolve}
\end{claim}

\subsection{Bounding the Norm of $\qqbar_i - \qqbar_{\sigma_i(j)}$ in terms of $\dd_i$'s}

Rearranging the above equation gives that for each $-1 \le i \le \ns-3, 1 \le j \le \min\{i+2,3\}$,
\begin{align}
\qqbar_i - \qqbar_{\sigma_i(j)}
= \dd_i - \frac{1}{\det(\FF_i)} \sum_{\substack{1 \le j' \le \min\{i+2,3\} \\ j' \neq j}} \HH_{i,j'}\left( \qqbar_{\sigma_i(j)} - \qqbar_{\sigma_i(j')} \right).
\label{eqn:recursion}
\end{align}
Our goal is to express each $\qqbar_i - \qqbar_{\sigma_i(j)}$ as a function of $\dd_i, \dd_{i-1}, \ldots, \dd_{-1}$.
If each term $\qqbar_{\sigma_i(j)} - \qqbar_{\sigma_i(j')}$ in the right hand side of Equation~\eqref{eqn:recursion} satisfies 
\[
\{ \sigma_i(j), \sigma_i(j') \} 
= \{ i_1, \sigma_{i_1}(j_1) \}
\]
for some $i_1 < i$ and $1 \le j_1 \le \min\{i_1+2,3\}$,  then we can substitute it by Equation~\eqref{eqn:recursion} with the left hand side being $\qqbar_{i_1} - \qqbar_{\sigma_{i_1}(j_1)}$.
The substitution terminates when the right hand side term becomes $\pm \left(\qqbar_{-1} - \qqbar_{-2} \right) = \pm \dd_{-1}$.

\begin{claim}
Let $\qqbar_{\sigma_{i_1}(j_1)} - \qqbar_{\sigma_{i_1}(j_2)}$ ($0 \le i_1 \le \ns-3, 1 \le j_1 \neq j_2 \le \min\{i+2,3\}$) be a term appearing in the right hand side of Equation~\eqref{eqn:recursion}.
There exist $-1 \le i_2 < i_1, 1 \le j_3 \le \min\{i_2+2, 3\}$ satisfying
\[
\{ \sigma_{i_1}(j_1), \sigma_{i_1}(j_2) \} 
= \{ i_2, \sigma_{i_2}(j_3) \}.
\]
\label{clm:recursion}
\end{claim}
\begin{proof}
Without the loss of generality, assume $\sigma_{i_1} (j_1) > \sigma_{i_1} (j_2)$.
If $i_1 = 0$, then 
$\{ \sigma_{i_1}(j_1), \sigma_{i_1}(j_2) \} = \{-1, \sigma_{-1}(1)\}$; 
if $i_1 = 1$, then
$\{ \sigma_{i_1}(j_1), \sigma_{i_1}(j_2) \} \in \{ \{0, \sigma_{0}(1)\}, \{0, \sigma_{0}(2)\}, \{-1, \sigma_{0}(1)\} \}$.
The remaining proof focuses on $i_1 \ge 2$.

\begin{figure}[htb]
\centering
\includegraphics[scale=0.55]{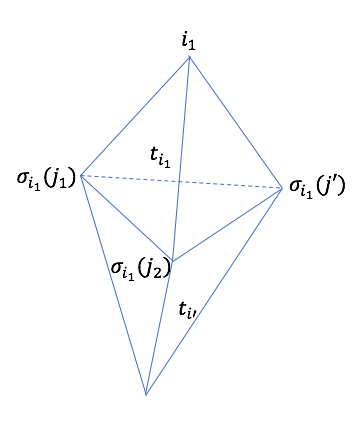}
\caption{Tetrahedrons $\tetrahedron_{i_1}$ and $\tetrahedron_{i'}$}
\label{fig:recursion}
\end{figure}

Let $\sigma_{i_1}(j')$ be the other vertex in tetrahedron $\tetrahedron_{i_1}$.
Let $\tetrahedron_{i'}$ be the parent of $\tetrahedron_{i_1}$ in the BFS tree $\Tbfs$.
$i' < i_1$, and the two tetrahedrons
$\tetrahedron_{i_1}$ and $\tetrahedron_{i'}$ share a triangle face containing vertices $\sigma_{i_1}(j_1), \sigma_{i_1}(j_2), \sigma_{i_1}(j')$.
See Figure~\ref{fig:recursion}.
If $i' = 1$,
then 
\[
\{\sigma_{i_1}(j_1), \sigma_{i_1}(j_2) \}
\in \{ \{i_2, \sigma_{i_2}(j_3)\}: -1 \le i_2 \le 1, 1 \le j_3 \le \min\{i_2+2,3\} \}.
\]
Otherwise $i' \ge 2$, we prove the statement by case analysis. \\
{\bf Case 1.} $\sigma_{i_1}(j_1) = \max_{1\le k \le 3}\{\sigma_{i_1}(k)\}$.
By our labeling rules, vertex $i'$ is the one with the maximum index among all the vertices of $\tetrahedron_{i'}$ and it is contained in the triangle shared by $\tetrahedron_{i_1}$ and  $\tetrahedron_{i'}$.
Thus, $\sigma_{i_1}(j_1) = i'$ and
\[
\{\sigma_{i_1}(j_1), \sigma_{i_1}(j_2) \}
= \{i', \sigma_{i'}(j_3) \}
\]
for some $1 \le j_3 \le 3$.\\
{\bf Case 2.} $\sigma_{i_1}(j') = \max_{1\le k \le 3}\{\sigma_{i_1}(k)\}$.
Then, $\sigma_{i_1}(j') = i'$ and 
$\sigma_{i_1}(j_1), \sigma_{i_1}(j_2) \in \{ \sigma_{i'}(k): 1\le k \le 3\}$.
Note $i' < i$.
By induction on the tetrahedron index $i'$, we can see that there exist $i_2, j_3$ satisfying
$\{ \sigma_{i_1}(j_1), \sigma_{i_1}(j_2) \} 
= \{ i_2, \sigma_{i_2}(j_3) \}$.
\end{proof}

We use a recursion tree $\Trec^{(i,j)}$ to express the process of recursively substituting $\qqbar_{i'} - \qqbar_{\sigma_{i'}(j')}$ via Equation~\eqref{eqn:recursion}.
The root of $\Trec^{(i,j)}$ is $\qqbar_i - \qqbar_{\sigma_i(j)}$, and 
each node $u$ of $\Trec^{(i,j)}$ represents a term $\qqbar_{i_u} - \qqbar_{\sigma_{i_u}(j_u)}$ which is substituted at that point.
The leaves are $\pm(\qqbar_{-1}-  \qqbar_{-2})$, which equals $\pm \dd_{-1}$ by Equation~\eqref{eqn:diDef}.

\begin{claim}
The number of nodes in the recursion tree $\Trec^{(i,j)}$ is at most $2^i$.
\label{clm:recursionTreeNodes}
\end{claim}
\begin{proof}
Since $\Trec^{(i,j)}$ is a binary tree,
it suffices to prove that the height of $\Trec^{(i,j)}$ is at most $i$.
For each non-leaf node $u$ of $\Trec^{(i,j)}$,
let $u_1$ be a child of $u$.
According to Equation~\eqref{eqn:recursion} and our labeling rules, 
$i_{u_1}$ is a vertex in tetrahedron $\tetrahedron_{i_u}$ and $i_{u_1} < i_u$.
Let $\tetrahedron_k$ be the tetrahedron such that $\tetrahedron_k$ and $\tetrahedron_{i_u}$
share a triangle face containing vertices $\sigma_{i_u}(1), \sigma_{i_u}(2), \sigma_{i_u}(3)$.
Then in $\Tbfs$, $\tetrahedron_k$ is the parent of $\tetrahedron_{i_u}$.
Since $i_{u_1} < i_u$, in $\Tbfs$, the depth of $\tetrahedron_{i_{u_1}}$ is smaller than the depth of $\tetrahedron_{i_u}$. 
See Figure~\ref{fig:recHeight}.
It implies that the height of $\Trec^{(i,j)}$ is at most the height of $\Tbfs$, which is at most $i$.
\begin{figure}[htb]
\centering
\includegraphics[scale=0.55]{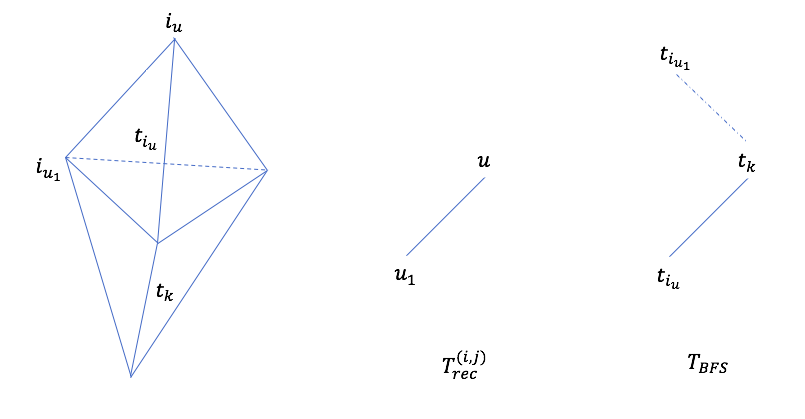}
\caption{If $u_1$ is a child of $u$ in $\Trec^{(i,j)}$, then the depth of $\tetrahedron_{i_{u_1}}$ is smaller than the depth of $\tetrahedron_{i_u}$ in $\Tbfs$.}
\label{fig:recHeight}
\end{figure}
\end{proof}

At each node $u$ of the recursion tree $\Trec^{(i,j)}$, 
by applying Equation~\eqref{eqn:recursion} with the left hand side being $\qqbar_{i_u} -\qqbar_{\sigma_{i_u}(j_u)}$, we introduce a term related to $\dd_{i_u}$.
For each non-root node $u$, denote $\pa(u)$ the parent of $u$ in $\Trec^{(i,j)}$.
Let $P_{u}$ be the path from the root node to node $u$ in $\Trec^{(i,j)}$.
For each non-root node  $u$, define
\begin{align}
\MM_u \defeq \prod_{v \in P_{pa(u)}} \frac{\HH_{i_v,j_v}}{\det(\FF_{i_v})}.
\label{eqn:muDef}
\end{align}
For root node $r$, define $\MM_r \defeq \II$.
By recursively applying Equation~\eqref{eqn:recursion}, we get
\[
\qqbar_i - \qqbar_{\sigma_i(j)}
= \sum_{u \in \Trec^{(i,j)}} \sgn{u} \MM_u \dd_{i_u},
\]
where $\sgn{u} \in \{+1,-1\}$.
By the triangle inequality and the multiplicative inequality of 2-norm,
\begin{align}
\norm{\qqbar_i - \qqbar_{\sigma_i(j)}}_2
\le \sum_{u \in \Trec^{(i,j)}} \norm{ \MM_u}_2 \norm{\dd_{i_u}}_2.
\label{eqn:xiDiff}
\end{align}

We bound $\norm{\MM_u}_2$ for each non-root node $u$ by the following claim.
\begin{claim}
Let $l_{\max}$ be the maximum edge length, and let $V_{\min}$ be the minimum tetrahedron volume. 
Define
$c \defeq \frac{l_{\max}^3}{6 V_{\min}}$, which is a constant by our assumption that all tetrahedrons have constant aspect ratios and all edge lengths are constant.
For each non-root $u$ in $\Trec^{(i,j)}$, $\norm{\MM_u}_2 \le c^i$. 
\label{clm:muBound}
\end{claim}
\begin{proof}
Let $i = u_1 > u_2 > \ldots > u_t > u_{t+1} = u$ be the vertices along the path $P_u$. 
By the proof of Claim~\ref{clm:recursionTreeNodes}, the  length of $P_u$ is at most $i$.
Let $i_k = i_{u_k}$ and $j_k = j_{u_k}$ for each $1 \le k \le t$.
By Equations~\eqref{eqn:muDef}, \eqref{eqn:HijDef}, and~\eqref{eqn:detFi},
\begin{align*}
\MM_u &= \prod_{i_1 \ge i_l \ge i_t} \frac{\hh_{i_{l},j_{l}} (\pp_{i_l} - \pp_{\sigma_{i_l}(j_l)})^\top }{(\pp_{i_l} - \pp_{\sigma_{i_l}(j_l)})^\top \hh_{i_l,j_l}} \\
& = \hh_{i_1,j_1} 
\left( \prod_{i_1 \ge i_l \ge i_{t-1}}
\frac{ (\pp_{i_l} - \pp_{\sigma_{i_l}(j_l)})^\top \hh_{(i_{l+1},j_{l+1})} }{(\pp_{i_l} - \pp_{\sigma_{i_l}(j_l)})^\top \hh_{(i_l,j_l)} } \right)
\frac{(\pp_{i_{t}} - \pp_{\sigma_{i_l}(j_t)})^\top }{(\pp_{i_{t}} - \pp_{\sigma_{i_l}(j_t)})^\top \hh_{(i_t,j_t)} }.
\end{align*}
Recall that from Claim~\ref{clm:determinant},
\[
(\pp_{i_l} - \pp_{\sigma_{i_l}(j_l)})^\top \hh_{(i_l,j_l)}
= \det(\FF_{i_l}), 
\] 
which is equal to the volume of the tetrahedron generated by vectors $\pp_{i_l} - \pp_{\sigma_{i_l}(1)}, \pp_{i_l} - \pp_{\sigma_{i_l}(2)}$, and $\pp_{i_l} - \pp_{\sigma_{i_l}(3)}$.
Similarly, $(\pp_{i_l} - \pp_{\sigma_{i_l}(j_l)})^\top \hh_{(i_{l+1},j_{l+1})}$ equals to the volume of the tetrahedron generated by vectors $\pp_{i_l} - \pp_{\sigma_{i_l}(j_l)}, \pp_{i_{l+1}} - \pp_{\sigma_{i_{l+1}}((j_{l+1}+1)_3)}$, and $\pp_{i_{l+1}} - \pp_{\sigma_{i_{l+1}}((j_{l+1}+2)_3)}$.
Thus, by our definition $c = l_{\max}^3 / (6 V_{\min}) \ge V_{\max} / V_{\min}$,
\[
\abs{\frac{ (\pp_{i_l} - \pp_{\sigma_{i_l}(j_l)})^\top \hh_{(i_{l+1},j_{l+1})}}{(\pp_{i_l} - \pp_{\sigma_{i_l}(j_l)})^\top \hh_{(i_l,j_l)}^{(j_l)}}} \le c.
\]
Thus,
\begin{align*}
\norm{\MM_u}_2 & 
\le  \left( \prod_{i_1 \ge i_l \ge i_{t-1}}
\abs{ \frac{ (\pp_{i_l} - \pp_{\sigma_{i_l}(j_l)})^\top \hh_{(i_{l+1},j_{l+1})}}{(\pp_{i_l} - \pp_{\sigma_{i_l}(j_l)})^\top \hh_{(i_{l},j_{l})}}} \right)
\norm{ \frac{\hh_{(i_{1},j_{l})} (\pp_{i_t} - \pp_{\sigma_{i_t}(j_t)})^\top }{(\pp_{i_{t}} - \pp_{\sigma_{i_t}(j_t)})^\top \hh_{(i_{t},j_{t})} } }_2 
= O(c^i).
\end{align*}
\end{proof}

Applying Claim~\ref{clm:muBound} to 
Equation~\eqref{eqn:xiDiff}, 
\[
\norm{\qqbar_i - \qqbar_{\sigma_i(j)}}_2
= O \left( c^i \sum_{u \in \Trec^{(i,j)}} \norm{\dd_{i_u}}_2 \right).
\]
By Claim~\ref{clm:recursionTreeNodes} and the fact that $-1 \le i_u \le i$ for each $u \in \Trec^{(i,j)}$,
\begin{align}
\norm{\qqbar_i - \qqbar_{\sigma_i(j)}}_2
= O \left( (2c)^i \sum_{-1 \le i' \le i} \norm{\dd_{i'}}_2 \right).
\label{eqn:xiNorm}
\end{align}

\subsection{Bounding the Quadratic Form $(\qqbar)^\top \MM \qqbar$}

Note
\[
(\qqbar)^\top \MM \qqbar
\ge \sum_{\substack{-1 \le i \le \ns-3 \\ 1 \le j \le \min\{i+2,3\}}}
\left( (\qqbar_i - \qqbar_{\sigma_i(j)})^\top (\pp_i - \pp_{\sigma_i(j)}) \right)^2.
\]
By Equation~\eqref{eqn:diReplace},
\[
(\qqbar)^\top \MM \qqbar
\ge \sum_{\substack{-1 \le i \le \ns-3 \\ 1 \le j \le \min\{i+2,3\}}}
\left( \dd_i^\top (\pp_i - \pp_{\sigma_i(j)}) \right)^2.
\]
By the fact that $\dd_{-1} = \qqbar_{-1} - \qqbar_{-2}$ is parallel to $\pp_{-1} - \pp_{-2}$, 
\[
\left( \dd_{-1}^\top (\pp_{-1} - \pp_{-2}) \right)^2
= \Omega \left( \norm{\dd_{-1}}_2^2 \right).
\]

\begin{claim}[Lemma 3.7 of~\cite{daitchS07}]
Under the assumption: the angle between $\pp_0- \pp_{-1}$ and $\pp_0 - \pp_{-2}$ is in the range $[\theta, \pi - \theta]$ for some constant $\theta$.
For any fixed unit vector $\yy \in \mathbb{R}^3$,
\[
\sum_{j \in \{-2,-1\}} \left( \yy^\top (\pp_0 -\pp_j) \right)^2 = \Theta(1).
\]
\label{clm:projectSum2}
\end{claim}

\begin{claim}
Under the assumption of constant edge lengths: for each $1 \le i \le \ns-3$, the determinant of the matrix:
\[
\det\left( \left( \begin{array}{ccc}
\pp_i - \pp_{\sigma_i(1)} & \pp_i - \pp_{\sigma_i(2)}
& \pp_i - \pp_{\sigma_i(3)}
\end{array} \right) \right)
= \Theta(1).
\]
Furthermore, for any fixed unit vector $\yy \in \mathbb{R}^3$, for each $1 \le i \le \ns-3$,
\[
\sum_{1 \le j \le 3} \left( \yy^\top \left(\pp_i - \pp_{\sigma_i(j)} \right) \right)^2 = \Theta(1).
\]
\label{clm:projectSum}
\end{claim}
\begin{proof}
The determinant of matrix $ \left( \begin{array}{ccc}
\pp_i - \pp_{\sigma_i(1)} & \pp_i - \pp_{\sigma_i(2)}
& \pp_i - \pp_{\sigma_i(3)}
\end{array} \right)$
is the signed volume of tetrahedron $\tetrahedron_i = \{i, \sigma_i(1), \sigma_i(2), \sigma_i(3)\}$, 
which is constant by our assumption.

We claim that for any unit vector $\yy$, there exists some $j \in \{1,2,3\}$ such that $\yy^\top (\pp_i - \pp_{\sigma_i(j)})= \Theta(1)$.
Assume by contradiction, 
for every $j \in \{1,2,3\}$ we have $\yy^\top (\pp_i - \pp_{\sigma_i(j)})= o(1)$.
This means that all three vectors $\pp_i - \pp_{\sigma_i(1)}, \pp_i - \pp_{\sigma_i(2)}, \pp_i - \pp_{\sigma_i(3)}$
are between two 2D planes with distance $o(1)$ and orthogonal to $\yy$.
This contradicts the assumption that the volume of tetrahedron $t_i = \{i, \sigma_i(1), \sigma_i(2), \sigma_i(3)\}$ is constant.

Thus, we have
\[
\sum_{1 \le j \le 3} \left( \yy^\top \left(\pp_i - \pp_{\sigma_i(j)} \right) \right)^2 = \Theta(1).
\]
\end{proof}

Therefore,
\[
(\qqbar)^\top \MM \qqbar
= \Omega \left( \sum_{-1 \le i \le n} \norm{\dd_i}_2^2 \right).
\label{eqn:quadraticBound}
\]

Applying the Cauchy-Schwarz inequality on the above equation and Equation~\eqref{eqn:xiNorm} gives the path lemma.

\def\dist{\text{dist}}
\def\vol{\text{vol}}
\def\algKChunksND{\textsc{ConvexTrussUnionND}}
\def\algBoundBox{\textsc{BoundingBox}}
\def\surf{\text{Surf}}

\def\calB{\mathcal{B}}

\section{Computing a $(B, r)$-Hollowing of a
Convex Edge-simple  Truss}
\label{sec:appendixHollow}

In this section, we present Algorithm~\ref{alg:hollow} $\textsc{Hollow}$, 
which is used as a subroutine for line~\ref{line:MainHollow} of Algorithm~\ref{alg:TrussSolver} $\textsc{TrussSolver}$.
The input of Algorithm~\ref{alg:hollow} consists of a {\trussnice}  3D truss $\calT = \langle \{\pp_i\}_{i \in V}, T, E, \gamma \rangle$, a bounding box $B$ and an integer parameter $r$.
The algorithm outputs a $(B, r)$-hollowing $\calH$ of $\calT$. We can check that the algorithm terminates in time $O(\abs{V})$.
These together prove Lemma~\ref{lem:hollow}.

\begin{algorithm}[htb]
\caption{
  $\textsc{Hollow}(
\calT = \langle \{\pp_i\}_{i \in V}, T, E, \gamma \rangle, B,
 r )$
  \label{alg:hollow}
}
\begin{algorithmic}[1]
\renewcommand{\algorithmicrequire}{\textbf{Input:}}
\renewcommand\algorithmicensure {\textbf{Output:}}
\REQUIRE{a {\trussnice}  3D truss $\calT = \langle \{\pp_i\}_{i \in V}, T, E, \gamma \rangle$, 
a bounding box $B$,
a positive integer $r \le n / \alpha(\calT)^2$}
\ENSURE{a $(B, r)$-hollowing of $\calT$}
\STATE Compute $r$-division planes, which are planes orthogonal
to each of the three directions of $B$
such that these planes divide $B$ into $O(n/r)$ small cubes of side length $O(r^{1/3})$ each. 
\label{lin:smallCubes}
\STATE $\calH \leftarrow$ tetrahedrons intersecting an $r$-division plane.\label{lin:startConstructH}
\STATE $\calH \leftarrow \calH \ \cup$ tetrahedrons on the boundary of $\calT$.
\STATE Make $\calH$ stiffly-connected by adding a minimal number of tetrahedrons in $\calT$.
\label{lin:endConstructH}
\RETURN $\calH$.
\end{algorithmic}
\end{algorithm}

Given that each individual tetrahedron has constant volume and constant aspect ratio, 
the following observation converts counting the number of tetrahedrons in a truss intersecting a 2D plane into the intersection area of this truss and the plane.

\begin{observation}
Let $\calT$ be a {\trussnice} 3D truss, and let $P$ be a 2D plane.
Let $A(\calT \cap P)$ be the intersection area of $\calT$ and $P$.
Then the number of tetrahedrons in $\calT$ intersecting $P$ is upper bounded by $O(\max\{A(\calT \cap P), 1\})$.
\label{obs:convertToSurfaceArea}
\end{observation}
\begin{proof}
Since every individual tetrahedron in $\calT$ has constant volume and constant aspect ratio, a tetrahedron of $\calT$ intersects $P$ only if all points of this tetrahedron is within some constant distance of $P$.
The number of tetrahedrons of $\calT$ intersecting $P$ can be upper bounded by the volume within some constant distance to $\calT \cap P$.
Thus, the number of tetrahedrons of $\calT$ intersecting $P$ is at most $O(\max\{A(\calT \cap P), 1\})$.
\end{proof}

\subsection{Bounding the Size of $\calH$}

In this section, we show that $\calH = \textsc{Hollow}(\calT, B, r)$, computed by Algorithm~\ref{alg:hollow}, has a small size. That is, $\calH$ satisfies the first condition of the $(B, r)$-hollowing definition in Definition~\ref{def:hollow}.

Note in Algorithm~\ref{alg:hollow} line~\ref{lin:smallCubes}, the bounding box $B$ is divided into $O(n/r)$ small cubes of volume $O(r)$ each.
We call a small cube as a \emph{region}.
The tetrahedrons in a region are the tetrahedrons of $\calH$ which intersect a single small cube.
A tetrahedron can appear in at most eight regions.

The following lemma upper bounds the number of tetrahedrons of $\calH$ in each region.

\begin{lemma}
Given a {\trussnice} 3D truss $\calT$ of $n$ vertices, a bounding box $B$ and a positive integer $r \le n / \alpha(\calT)^2$, 
let $\calH = \textsc{Hollow}(\calT, B, r)$ returned by Algorithm~\ref{alg:hollow}. 
Then, $\calH$ has at most $O(nr^{-1/3})$ tetrahedrons.
\label{lem:regionBoundary} 
\end{lemma}
Note the shortest side length of the bounding box of $\calT$ is at least $n^{1/3} \alpha^{-2/3}$. 
The requirement $r \le n / \alpha(\calT)^2$ guarantees that the shortest side length is at least $r^{1/3}$ so that an $r$-division exists.

\begin{proof}

Note $\calH$ has $O(n/r)$ regions.
It suffices to show that each region of $\calH$ has at most $O(r^{2/3})$ tetrahedrons.

Let $R$ be a region of $\calH$.
A tetrahedron of $\calH$ belongs to region $R$ if either this tetrahedron is
within constant distance to the boundary of $R$,
or this tetrahedron is within constant distance to the part of
the boundary of $\calT$ that's contained in $R$.
Since every tetrahedron of $\calH$ has constant volume and aspect ratio,
the number of tetrahedrons within constant distance to
the boundary of $R$ is $O(r^{2/3})$.
It remains to bound the number of tetrahedrons within constant
distance to the boundary of $\calT$ that's contained in $R$.

Define $S \defeq \calT \cap R$.
Since both $\calT$ and $R$ are convex, $S$ is convex.
Let $\surf(\cdot)$ the surface area of a shape.
By Observation~\ref{obs:convertToSurfaceArea}, the number of boundary tetrahedrons of $\calT$ contained in $R$ is $O(\surf(S))$.

Let $B_1$ be the smallest ball containing $S$.
Since $R$ is a cube, we have
\[
\surf(B_1) = O(\surf(R)).
\]
So it suffices to show $\surf(S) \leq \surf(B_1)$.
We do so by giving a one-to-one mapping of every
point from the surface of $S$ onto $B_1$.

Let $\phi: S \rightarrow B_1$ 
which maps each face of $S$ to a subset of the surface of $B_1$, defined as follows.
Consider a face of $S$, say $f$, with vertices
$v_1, \ldots, v_k$ in a clockwise order.
Let $P_f$ be the plane containing $f$.
$P_f$ cuts $B_1$ into two parts, let $B_1'$ be the part of the smaller volume
(break a tie arbitrarily), aka the sphere cap generated by the plane $p_f$.

For each $1 \le i \le k$, let $P_i$ be the plane orthogonal to $P_f$
that passes through $v_i$ and $v_{i+1})$
(if $i =k$, then the intersection line is $(v_k, v_1)$).
Let $u_i$ be the point of intersection of the surface of $B_1$
with the planes $P_{i - 1}$ and $P_i$
(if $i = 1$, then $u_1$ is the intersect vertex of $P_1, P_k$ and the surface of $B_1'$).

We define $\phi(f)$ to be the surface  of $B_1'$ enclosed by $(u_1, \ldots, u_k, u_1)$.
See Figure~\ref{fig:surfaceArea} for an example.
\begin{figure}[htb]
\centering
\includegraphics[width=0.5\textwidth]{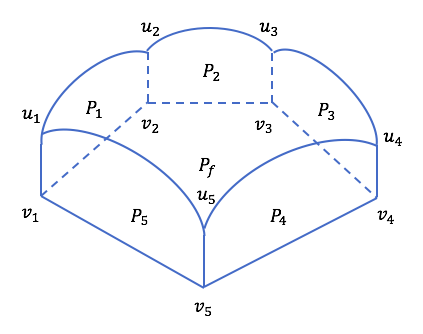}
\caption{An example of $\phi(f)$. Face $f$ has five vertices $v_1, \ldots, v_5$. $\phi(f)$ is the surface of $B_1'$ enclosed by $(u_1, \ldots, u_5, u_1)$.}
\label{fig:surfaceArea}
\end{figure}

For each face $f$ of $S$, the orthogonal projection
of $\phi(f)$ onto the plane $P_f$ is $f$.
Thus
\[
\textsc{Area}(\phi(f)) \ge \textsc{Area}(f).
\]
In addition, since $S$ is convex, for any two distinct
faces $f_1$ and $f_2$, $\phi(f_1)$ and $\phi(f_2)$ are also disjoint.
So we have
\[
\text{Surf}\left(S\right)
= \sum_{f \in S} \textsc{Area}\left(f\right)
\leq \sum_{f\in S} \textsc{Area}\left(\phi\left(f\right)\right)
\leq \text{Surf}\left(B_1\right).
\]
Combining this with $\text{Surf}(B_1) \leq O(\text{Surf}(S)) \leq O(r^{2/3})$
then completes the proof.
\end{proof}

\subsection{Bounding the Number of Tetrahedrons in $\calH$ Intersecting with a Plane}

In this section, we show that $\calH = \textsc{Hollow}(\calT, B, r)$, computed by Algorithm~\ref{alg:hollow}, has a small overlap with any plane whose normal vector has an angle between $(0,\pi/2)$ with the longest direction of $B$.
That is, $\calH$ satisfies the second condition of the $(B, r)$-hollowing definition in Definition~\ref{def:hollow}.

\begin{lemma}
Given a {\trussnice} 3D truss $\calT$ of $n$ vertices, a bounding box $B$ and a positive integer $r \le n / \alpha(\calT)^2$, let $\calH = \textsc{Hollow}(\calT, B, r)$  returned by Algorithm~\ref{alg:hollow}.
Let $\dd \in \mathbb{R}^3$ be a unit vector such that the angle
between $\dd$ and let the angles with the three directions of the box
(normals to its faces) be $\theta_x, \theta_y, \theta_z \in (\theta,
\pi/2)$, for some $\theta > 0$.
Then, the number of tetrahedrons in $\calH$ which intersect any plane $P$ orthogonal to $\dd$ is at most 
\[
O \left( \frac{n^{2/3} }{\alpha(\calT)^{1/3}  r^{1/3} \cos^2 \theta}
\right).
\]
\label{lem:planeIntersect}
\end{lemma}

Without loss of generality, we assume that the bounding box $B$ is \emph{axis-parallel}, that is, the sides of $B$ are parallel to the three axes: the $x$ axis, the $y$ axis and the $z$-axis.
We say an axis-paralleled box has side lengths $a, b, c > 0$, if the sides parallel to the $x$-axis have length $a$, the sides parallel to the $y$-axis have length $b$, and the sides parallel to the $z$-axis have length $c$, respectively.

To prove Lemma~\ref{lem:planeIntersect}, we need the following claim, which bounds the intersection area of a 2D plane and a 3D box.

\begin{claim}
Let $B$ be a 3D axis-parallel box of side lengths $a, b, c > 0$. 
Let $\dd \in \mathbb{R}^3$ be a unit vector such that the angle between $\dd$ and the $x$-axis (the $y$-axis, and the $z$-axis) is $\theta_x$ ($\theta_y, \theta_z$, respectively).
Suppose $\theta_x, \theta_y, \theta_z \in (0, \pi/2)$.
Then the intersection area $S$ of any 2D plane $P$ orthogonal to $\dd$ and the 3D box $B$ satisfies 
\[
S \le \min \left\{ \frac{bc}{\cos \theta_x}, \frac{ac}{\cos \theta_y}, \frac{ab}{\cos \theta_z} \right\}.
\]
\label{clm:intersectArea}
\end{claim}

\begin{proof}

Since the three terms in the right hand side are symmetric, we only prove $S$ can be upper bounded by the first term and the other two follow in a similar way.

Let $B_{x1}, B_{x2}$ be the two faces of $B$ which are orthogonal to the $x$-axis, without loss of generality, assume $B_{x1}$ has a smaller $x$-coordinate.

If $P$ intersects neither $B_{x1}$ nor $B_{x2}$, then the volume of $B$ is equal to $S a \cos \theta_x$.

If $P$ intersects $B_{x1}$ say with line $(K,L)$, see Figure~\ref{fig:intersectArea}, then we draw a line going through point $K$ and parallel to the $x$-axis, which intersects face $B_{x2}$ at point $M$, similarly we draw a line going through point $L$ and parallel to the $x$-axis, which intersects face $B_{x2}$ at point $N$.
We cut the box $B$ by the plane $KMLN$, 
see Figure~\ref{fig:intersectArea}.
Note that the volume of the convex hull of $(K,G,E,F,G,H,I,J)$, the right one in Figure~\ref{fig:intersectArea}, equals to $S a \cos \theta_x$, which is smaller than the volume of $B$.
\begin{figure}[htb]
\centering
\includegraphics[width=\textwidth]{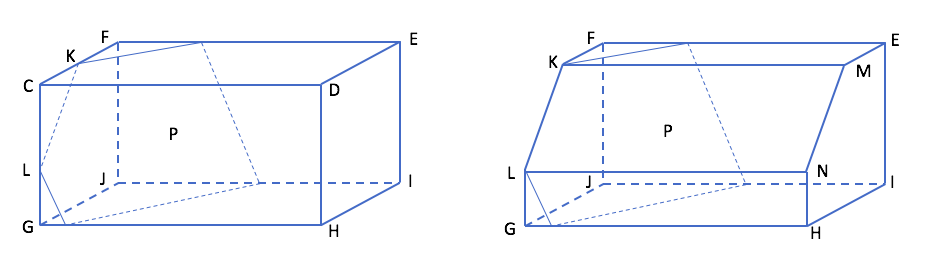}
\caption{Plane $P$ intersects face $B_{x1}$ with line $(K,L)$. We draw a line going through point $K$ and parallel to the $x$-axis, which intersects face $B_{x2}$ at point $M$, similarly we draw a line going through point $L$ and parallel to the $x$-axis, which intersects face $B_{x2}$ at point $N$.}
\label{fig:intersectArea}
\end{figure}
\begin{figure}[htb]
\centering
\includegraphics[width=0.6\textwidth]{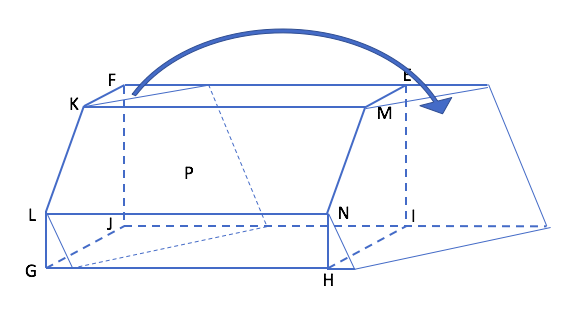}
\caption{Cut the shape along the plane $P$, and shift the left part along the $x$-axis and glue the two faces $KFJGL$ and $MEIHN$.}
\label{fig:intersectArea2}
\end{figure}

Similarly, if $P$ intersects $B_{x2}$,  
then we can draw two lines parallel to the $x$-axis and going through the two intersection points respectively and get a shape of volume $Sa\cos\theta_x$ smaller than the volume of $B$.
Note that the intersection between $P$ and $B_{x2}$ cannot coincide with $(M,N)$, otherwise the angle $\theta_x = \pi / 2$.
We can check that the shape we get must contain $P \cap B$, and the two faces of the shape which are orthogonal to the $x$-axis are congruent.
If we cut the shape along the plane $P$, then we can shift the left part along the $x$-axis and glue the two faces which are orthogonal to the $x$-axis, and get a parallelepiped which has $P \cap B$ as a face. Figure~\ref{fig:intersectArea2} shows the parallelepiped we get from the example of Figure~\ref{fig:intersectArea}.
Thus,
\[
S \cdot a \cdot \cos \theta_x \le abc.
\]
That is, $S \le bc / \cos \theta_x$.
This completes the proof.
\end{proof}

Now we prove Lemma~\ref{lem:planeIntersect}.

\begin{proof}[Proof of Lemma~\ref{lem:planeIntersect}]

Let $t$ be the number of regions (that is, small cubes of the hollowing) of $\calH$ which intersect a 2D plane $P$,
and let $m$ be the maximum number of tetrahedrons of a single region of $\calH$ which intersect the plane $P$.
The number of tetrahedrons in the hollowing $\calH$ which intersect  $P$ can be upper bounded by $tm$.

We first bound $t$, the number of regions of $\calH$ which intersect the plane $P$.
A cube region intersects the plane $P$ only if all its points are within a distance $\sqrt{3}r^{1/3}$ of $P$.
We put two planes, say $P_1, P_2$, which are parallel to the plane $P$, above and below $P$ with distance $\sqrt{3} r^{1/3}$ to $P$.
All cube regions which intersect the plane $P$ must be within the two planes $P_1$ and $P_2$.
By Claim~\ref{clm:intersectArea},
the volume of the box $B$ between the two planes $P_1, P_2$ is at most 
\[
\frac{bc}{\cos \theta} \cdot 2\sqrt{3} r^{1/3}.
\]
Since each cube has volume at most $r$, we can bound the number of  cube regions which intersect the plane $P$ 
\[
t \le \frac{2\sqrt{3} bc}{r^{2/3} \cos \theta}.
\]

We then bound $m$, the maximum number of tetrahedrons of a single region of $\calH$ which intersect the plane $P$.
Consider the tetrahedrons of $\calH$ in this single cube region which intersect a some hollowing plane.
These tetrahedrons are within constant distance of an $r^{1/3} \times r^{1/3}$ square, 
which is on a plane orthogonal to one of the direction of the bonding box.
By Observation~\ref{obs:convertToSurfaceArea} and Claim~\ref{clm:intersectArea}, the number of these tetrahedrons which intersect the plane $P$ is at most $O(r^{1/3} / \cos\theta)$.
Thus, we can bound the maximum number of tetrahedrons of a single  cube region which intersect the plane $P$
\[
m \le O\left( \frac{r^{1/3} }{ \cos\theta} \right).
\]
Therefore, the number of tetrahedrons of the hollowing $\calH$ which intersect the plane $P$ is at most
\[
tm \le O \left( \frac{bc}{r^{1/3} \cos^2 \theta} \right).
\]

Note that we have $a = c \cdot \alpha(\calT)$ and $a \ge b \ge c$. The above number is upper bounded by
\[
O \left( \frac{n^{2/3}}{\alpha(\calT)^{1/3} r^{1/3} \cos^2 \theta} \right).
\]
\end{proof}

\subsection{Bounding the Relative Condition Number of $\calH$}

In this section, we show that the condition number of $\AA_{\calH}$ and $\textsc{Sc}[\AA_{\calT}]_U$ is small.
That is, $\calH$ satisfies the third condition of the $(B, r)$-hollowing definition in Definition~\ref{def:hollow}.

Since $\calT$ is convex,
each region of $\calH$ is connected.
Lemma~\ref{lem:3dMinEig} implies that $\AA_{\calH}$ and $\textsc{Sc}[\AA_{\calT}]_U$ has the same null space.
By Lemma~\ref{lem:3dMinEig},
the smallest nonzero eigenvalue of $\AA_{\calH}$ is lower bounded by the diameter and the size of each region of $\calH$.

\begin{lemma}
Let $\AA_{\calT}$ be the truss stiffness matrix of a {\trussnice} 3D truss $\calT$.
Given a positive integer $r$ and a bounding box $B$, let 
$\calH = \textsc{Hollow}(\calT,B, r)$ be returned by Algorithm~\ref{alg:hollow}. 
Let $\AA_{\calH}$ be the associated truss stiffness matrix of $\calH$.
Then,
\[
\AA_{\calH} \pleq \textsc{Sc}[\AA_{\calT}]_U
\pleq O(r^2) \AA_{\calH},
\]
where $U$ consists of all vertices in $\calH$.
\label{lem:ConditionNumberBoundary}
\end{lemma}

\begin{proof}
The first inequality is equivalent to: $\textsc{Sc}[\AA_{\calT}]_U - \AA_{\calH}$ is a symmetric PSD matrix.
Let $\calT'$ be the truss obtained by removing all the edges of $\calT$ whose two endpoints are both in $U$.
Let $\AA_{\calT'}$ be the associated truss stiffness matrix of $\calT'$.
We can check that
\[
\textsc{Sc}[\AA_{\calT}]_U - \AA_{\calH}
= \textsc{Sc}[\AA_{\calT'}]_{U}.
\]
By Fact~\ref{fac:Schur}, $\textsc{Sc}[\AA_{\calT}]_U - \AA_{\calH}$ is a  symmetric PSD matrix.

It remains to prove the second inequality.
Note the planes of Algorithm~\ref{alg:hollow} line~\ref{lin:smallCubes} divide $\calT$ into small regions. 
Let $\calT_i$ denote the subgraph induced by $\calT$ on the $i$th region, and let $\calT_{C_i}$ denote the subgraph induced by $\calT$ on the boundary  of the $i$th region.
Let $\AA_{\calT_i}, \AA_{\calT_{C_i}}$ be the associated truss stiffness matrices of $\calT_i$ and $\calT_{C_i}$ respectively.
Let $U_i$ denote the boundary vertices of the $i$th region.
Let $\textsc{Sc}[\AA_{\calT_i}]_{U_i}$ be
the Schur complement of $\AA_{\calT_i}$ w.r.t. to $U_i$.

Since $\calT$ is a {\trussnice} 3D truss,
each vertex in $\calT_i$ has constant degree
and each edge 
has constant length and elasticity parameter.
It implies $\lambda_{\max}(\AA_{\calT_i}) = O(1)$.
By Fact~\ref{fac:Schur}, $\lambda_{\max}(\textsc{Sc}[\AA_{\calT_i}]_{U_i}) = O(1)$.

According to Algorithm~\ref{alg:hollow} from line~\ref{lin:startConstructH} to line~\ref{lin:endConstructH}, in each $\calT_{C_i}$, the tetrahedrons are arranged in simplicial complex and $\calT_{C_i}$ is connected.
By Lemma~\ref{lem:3dMinEig}, the null spaces of $\AA_{\calT_{C_i}}$ and the null space of $\textsc{Sc}[\AA_{\calT_i}]_{U_i}$ are the same.
Besides, each $\calT_{C_i}$ has $O(r^{2/3})$ vertices and diameter $O(r^{1/3})$, by Lemma~\ref{lem:regionBoundary}.
Applying Lemma~\ref{lem:3dMinEig} gives:
\[
\lambda_{\min}(\AA_{\calT_{C_i}}) = \Omega \left( \frac{1}{r^2} \right).
\]
This implies:
\[
\textsc{Sc}[\AA_{\calT_i}]_{U_i} \pleq O(r^2) \AA_{\calT_{C_i}}.
\]
Note each edge of $\calT$ only appears in a constant number of regions.
Thus,
\[
\textsc{Sc}[\AA_{\calT}]_{U} \pleq \sum_{i} \textsc{Sc}[\AA_{\calT_i}]_{U_i}
\pleq O(r^2) \sum_{i} \AA_{\calT_{C_i}}
\pleq O(r^2) \AA_{\calH}.
\]
This completes the proof.
\end{proof}

\section{Nested Dissection}
\label{sec:NestedDissection}

In this section,
we present our solver for the constructed preconditioners, Algorithm~\ref{alg:NDkChunks} $\algKChunksND$, which is used as a subroutine in line~\ref{line:MainND} of Algorithm~\ref{alg:TrussSolver} \textsc{TrussSolver}.

We first prove Lemma~\ref{lem:flexibleND} in Section~\ref{sec:flexibleNDProof}.
Then, in Section~\ref{sec:NDsubsection}, we use Lemma~\ref{lem:flexibleND}
to analyze the performance of Algorithm~\ref{alg:NDkChunks}, which proves Lemma~\ref{lem:kChunkND}.

\subsection{Proof of Lemma~\ref{lem:flexibleND}}
\label{sec:flexibleNDProof}

We restate Lemma~\ref{lem:flexibleND} with running time phrased in terms of the matrix
multiplication exponent $\omega$.

\begin{replemma}{lem:flexibleND}
Suppose we have a recursive separator decomposition
of a 
simplicial complex with $n$ bounded aspect ratio
tetrahedrons such that:
\begin{enumerate}
\item the number of leaves, and hence total number
of recursive calls, is at most $n^{\alpha}$.
\item each leaf (bottom layer partition) has at most
$n^{\beta}$ tetrahedrons.
\item each top separator has size at most $n^{\gamma}$.
\end{enumerate}
Then we can find an exact Cholesky factorization of
the associated stiffness matrix with 
multiplication count
$O(n^{\alpha + 2\omega \beta / 3 }  + n^{\alpha + \gamma \omega  }  \log n)$
and  fill-in size
$O(n^{\alpha + \frac{4}{3} \beta}  + n^{\alpha + 2 \gamma})$.
\end{replemma}

Nested dissection according to the separator decomposition stated in Lemma~\ref{lem:flexibleND} has three parts of cost:
\begin{enumerate}
\item Inverting leaf components:
	\begin{enumerate}
	\item the cost only associated with vertices not belonging to top-level separators; \label{item:leafOnlyCost}
	\item the cost associated with top-level separators. \label{item:leafSepCost}
	\end{enumerate}
\item Inverting top-level separators. \label{item:sepOnlyCost}
\end{enumerate}

We first analyze the cost of inverting top-level separators.
Note that top-level separators are disjoint.
After eliminating all leaf components, we get a \emph{layered graph}.
That is, its vertices can be partitioned into 
$V_1 \cup \ldots \cup V_l$ for some positive integer $l$ such that there is no edge between 
$V_i$ and $V_j$ for $\abs{i - j} > 1$.
Algorithm~\ref{alg:layerGraph} gives a numbering for a layered graph, and Claim~\ref{clm:layeredGraphND} analyzes its performance.

\begin{algorithm}[htb]
\caption{
  $\textsc{LayeredGraphND}(G = (V_1 \cup \ldots \cup V_l, E))$
  \label{alg:layerGraph}
}
\begin{algorithmic}[1]
\renewcommand{\algorithmicrequire}{\textbf{Input:}}
\renewcommand\algorithmicensure {\textbf{Output:}}
\REQUIRE{a layered graph $G$ with $l$ layers
}
\ENSURE{an elimination ordering for vertices in $G$
}

\IF{$l=1$}
	\STATE number the vertices in $V_1$ arbitrarily and return this numbering.
\ENDIF

\STATE Label $V_{\lfloor l/2 \rfloor}$ with the highest possible numbers.

\RETURN the numbering of $\textsc{LayeredGraphND}(G [V_1 \cup \ldots \cup V_{\lfloor l/2 \rfloor - 1}]$ and 
$\textsc{LayeredGraphND}(G [V_{\lfloor l/2 \rfloor + 1} \cup \ldots \cup V_l]$, and $V_{\lfloor l/2 \rfloor}$.

\end{algorithmic}
\end{algorithm}

\begin{claim}
Let $G$ be a layered graph with $l$ layers of at most $s$ vertices each.
Algorithm~\ref{alg:layerGraph} $\textsc{LayeredGraphND}$ returns a numbering such that, Gaussian elimination according to this order has fill-in size $O(s^2l)$ and multiplication count $O(s^{\omega} l)$.
\label{clm:layeredGraphND}
\end{claim}

\begin{proof}
We follow the proofs of~\cite{liptonRT79}.
We first bound the fill-in size.
Consider a recursion of Algorithm~\ref{alg:layerGraph} on a graph with $n$ vertices,
let $f(n)$ denote the maximum number of fill-in edges whose lower numbered endpoint is numbered by this recursion.
Suppose this recursion deals with layers $V_s, V_{s+1} \ldots, V_t$.
The algorithm numbers the vertices in the middle layer, that is, $V_{\lfloor(s+t)/2 \rfloor}$, which can be viewed as a separator whose removal separates the graph into two disjoint and balanced parts.
The fill-in edges whose lower numbered endpoint is in $V_{\lfloor(s+t)/2 \rfloor}$ consists of the following edges: (1) edges whose both endpoints are in $V_{\lfloor(s+t)/2 \rfloor}$; and (2) edges whose one endpoint is in $V_{\lfloor(s+t)/2 \rfloor}$ and the other is in $V_{s-1}$ (if exists) or $V_{t+1}$ (if exists).
Since each layer has at most $s$ vertices,
\[
f(n) \le \left\{ \begin{array}{ll}
O(s^2), &\quad \text{if } n = O(s) \\
f(n_1) + f(n_2) + 3s^2, & \quad \text{otherwise}
\end{array} 
\right.
\]
Note $n_2 \le n_1 \le n_2 + s$.
Thus, the total fill-in size is $O(s^2 l)$.

We then bound the multiplication count.
For a recursion of Algorithm~\ref{alg:layerGraph} on a graph with $n$ vertices, let $g(n)$ denote the maximum multiplication count associated with vertices in this graph which are going to be numbered in this recursion.
Similar to the analysis of the fill-in, we have
\[
g(n) \le \left\{ \begin{array}{ll}
O(s^{\omega}),  &\quad \text{if } n = O(s) \\
g(n_1) + g(n_2) + 3s^{\omega}, & \quad \text{otherwise}
\end{array} 
\right.
\]
Thus, the total multiplication count is $O(s^{\omega} l)$.
\end{proof}

Using Algorithm~\ref{alg:layerGraph} as a subroutine, we prove Lemma~\ref{lem:flexibleND}.

\begin{proof}[Proof of Lemma~\ref{lem:flexibleND}]

We first bound the fill-in size.
Note each leaf component has $O(n^{\beta})$ vertices in which $O(n^{\gamma})$ vertices are in some top-level separators  and are labeled by higher numbers.
By the result in~\cite{liptonRT79} and~\cite{millerT90}, the fill-in introduced by inverting 
 each leaf component is
\[
O\left( n^{4 \beta / 3} + n^{ 2 \beta / 3 + \gamma} \right).
\]
There are totally $O(n^{\alpha})$ leaf components.
Thus, the fill-in size introduced by inverting all leaf components is $O\left( n^{\alpha+ 4 \beta / 3} + n^{\alpha+ 2 \beta / 3 + \gamma} \right)$.

By Claim~\ref{clm:layeredGraphND},
the fill-in size of inverting top-level separators is $O(n^{\alpha + 2\gamma})$.
Thus, the total fill-in size is
\[
O\left(\alpha +  n^{4 \beta / 3  } + n^{\alpha +  2 \beta / 3 + \gamma }
+ n^{\alpha +  2\gamma} \right)
= O \left( n^{\alpha +  4 \beta / 3 } 
+ n^{\alpha + 2\gamma} \right).
\]
Here, we use Cauchy-Schwarz inequality to drop the second term.

We then bound the multiplication count.
By~\cite{millerT90}, when inverting each leaf component, the multiplication count only associated with vertices not belonging to top-level separators is
 $O(n^{2\omega \beta / 3})$.
By Claim~\ref{clm:layeredGraphND}, the multiplication count of inverting top-level separators is $O(n^{\alpha +  \omega \gamma})$.

It remains to upper bound the multiplication count associated with top-level separators when inverting a leaf component.
We adapt the analysis in~\cite{liptonRT79}.

We prove a more general result.
Consider a tetrahedral mesh $G$ with $n$ vertices, in which $s$ vertices are in a top-level separator and have been labeled with higher numbers.
Let $g(n, s)$ be the multiplication count associated these $s$ top-level vertices when running nested dissection on $G$.

Nested dissection finds a separator of size $n^{2/3}$ for $G$.
The multiplication count associated with the $s$ top-level vertices when inverting this separator can be upper bounded by $O((n^{2/3} + s)^\omega)$.
This gives the following recursion:
\begin{align*}
g(n, s) &\le O((n^{2/3} + s)^{\omega}) + \max_{n_i, s_i} 
\{\sum_{1\le i \le 2} g(n_i, s_i) \} \\
& \le O(2^\omega n^{2\omega /3} + 2^\omega s^{\omega} ) + \max_{n_i, s_i} \{ \sum_{1\le i \le 2} g(n_i, s_i) \}.
\end{align*}
Here, the maximum is taken over
\begin{align*}
s_1 + s_2 & \le s + n^{2/3} \\
n_1 + n_2 & \le n + n^{2/3}  \\
 n_1, n_2 &\le \frac{4}{5} n + n^{2/3}.
\end{align*}
We can compute that
\[
g(n, s) = O\left( n^{2\omega / 3} + s^\omega \log n\right).
\]

Note each leaf component is incident to at most two top-level separators.
Thus, the multiplication count associated with the top-level vertices when inverting this leaf component is
\[
O \left( n^{2\omega \beta / 3} + n^{\gamma \omega}  \log n \right).
\]
Combining the other two parts of cost, and the fact that there are totally $n^{\alpha}$ leaf components, 
the total multiplication count is
\[
O \left( n^{\alpha + 2\omega \beta / 3} + n^{\alpha +  \gamma \omega } \log n \right).
\]

If we replace $\omega$ by 3, then we can drop the $\log n$ factor in the above equation, which gives the result in the statement.
\end{proof}

\subsection{Proof of Lemma~\ref{lem:kChunkND}}
\label{sec:NDsubsection}

We now present Algorithm~\ref{alg:NDkChunks} $\algKChunksND$.
The input consists of a {\esimple} 3D truss $\calT$ which is
a union of  $k$ {\trussnice} trusses, a bounding box for each {\trussnice} truss,
the index subset of small-aspect-ratio trusses and $(B_i, r_i)$-hollowings for each small-aspect-ratio truss.
The output is an elimination ordering for the union of the hollowings of 
small-aspect-ratio trusses and the large-aspect-ratio trusses.
This proves Lemma~\ref{lem:kChunkND}.

\begin{algorithm}[htb]
\caption{
  $\algKChunksND(\calT = \langle \{\pp_i\}_{i \in V}, T, E, \gamma \rangle, \calB, \calI, \calH, l)$
  \label{alg:NDkChunks}
}
\begin{algorithmic}[1]
\renewcommand{\algorithmicrequire}{\textbf{Input:}}
\renewcommand\algorithmicensure {\textbf{Output:}}
\REQUIRE{a {\esimple} 3D truss $\calT$ which is the union of $k$ {\trussnice} trusses $\calT_1, \ldots, \calT_k$, \\
$\calB = \{B_1, \ldots, B_k \}$ in which $B_i$ is a bounding box of $\calT_i$
constructed from Line~\ref{line:BoundingBox} of Algorithm~\ref{alg:TrussSolver},\\
Index set $\calI$ of large aspect ratio complexes constructed from
Line~\ref{line:MainI} of Algorithm~\ref{alg:TrussSolver},\\
Hollowings of each $\calT_{i}$ with $i \in \calI$ constructed from
Line~\ref{line:MainHollow} of Algorithm~\ref{alg:TrussSolver},\\
$l$, the number of top-level partitions.
}
\ENSURE{an elimination ordering for vertices in $(\union_{i \in \calI} \calH_i) \bigcup (\union_{i \notin \calI} \calT_i)$
}

\STATE For each $1\le i \le k$, 
$\dd_i \leftarrow$ the direction of 
the longest sides of $B_i$.
\STATE Compute a unit vector $\dd$ such that $\frac{1}{10 k} \le \abs{\dd^\top \dd_i} \le 1- \frac{1}{10k}, \forall 1 \le i \le k$.\label{lin:directionD}
\STATE Compute separator planes $P_1, \ldots, P_l$, which are planes orthogonal to $\dd$ and dividing
$\calT$ into $l+1$ parts $Q_1, \ldots, Q_{l+1}$ of $O(n l^{-1})$ tetrahedrons each.
\label{lin:separatorPlane}

\STATE For each $1 \le  j\le l$,
$S_j \leftarrow$ tetrahedrons in $(\union_{i \in \calI} \calH_i) \bigcup (\union_{i \notin \calI} \calT_i)$ which intersect plane $P_j$. \label{lin:topSeparator}
\STATE $Q_1 \leftarrow Q_1 \cup S_1, Q_{l+1} \leftarrow Q_{l+1} \cup S_l$. 
\STATE For each $2 \le j \le l, Q_j \leftarrow Q_j \cup S_{j-1} \cup S_j$.

\STATE $\textsc{LayeredGraphND}(\calT[S_1 \cup \ldots \cup S_l])$ with highest numbers.
\label{lin:NDParallelSeparator} 

\STATE \label{lin:NDMT} Run nested dissection with  MT-separators\footnote{the separator algorithm in~\cite{millerT90}} for each $Q_j$ to number its unnumbered vertices.
\RETURN the elimination ordering of vertices in $(\union_{i \in \calI} \calH_i) \bigcup (\union_{i \notin \calI} \calT_i)$.
\end{algorithmic}
\end{algorithm}

We first prove that there exists a good direction $\dd$ such that:
the angle between $\dd$ and the longest direction of each bounding box is in a proper range.

\begin{lemma}
Let $k \ge 2$ and $\dd_1, \ldots, \dd_k \in \mathbb{R}^3$ be $k$ unit vectors. 
Then there exists a unit vector $\dd \in \mathbb{R}^3$ such that
\begin{align}
\frac{1}{10k}  \le \abs{\dd^\top \dd_i} \le 1- \frac{1}{10k}  , \forall 1 \le i \le k.
\label{eqn:directionD}
\end{align}
\end{lemma}
\begin{proof}
We pick a unit vector $\dd$ uniformly at random.
For any fixed $i$,
\[
\Pr \left( \frac{1}{10k} \le  \abs{\dd^\top \dd_i} \le 1 - \frac{1}{10k} \right)
=2 \cdot  \frac{\vol(1/10k) - \vol(1-  1/10k)}{V}.
\]
Here, 
$\vol(x)$ is the volume of a cap of a 3D unit ball with height $1-x$, and $V$ is the volume of a 3D unit ball.
We have $\vol(x) = \frac{1}{6} \pi (1-x^2)(3(1-x^4) + (1-x^2)^2)$ and $V = \frac{4}{3} \pi$.
Plugging these volumes into the above equation,
\begin{align*}
\Pr \left( \frac{1}{10k} \le  \abs{\dd^\top \dd_i} \le 1 - \frac{1}{10k} \right)
& = \frac{1}{2} \left( \left( 1 - \frac{1}{10k} \right)^2 
\left(2+ \frac{1}{10k} \right)
- \left(\frac{1}{10k} \right)^2 \left(3 - \frac{1}{10k} \right)
\right).
\end{align*}
Take the opposite:
\begin{align*}
\Pr \left( \abs{\dd^\top \dd_i} < \frac{1}{10k} \text{ or } \abs{\dd^\top \dd_i} > 1 - \frac{1}{10k} \right)
& \le  \frac{3}{20k} - \frac{1}{2000k^3}
+ \frac{3}{200k^2}.
\end{align*}
By union bound,
\[
\Pr \left( \exists i \ s.t. \ \abs{\dd^\top \dd_i} < \frac{1}{10k} \text{ or } \abs{\dd^\top \dd_i} > 1 - \frac{1}{10k} \right)
\le  k \left( \frac{3}{20k} - \frac{1}{2000k^3} + \frac{3}{200k^2} \right) 
\le \frac{1}{5}.
\]
Thus,
\[
\Pr \left( \forall i, \frac{1}{10k} \le  \abs{\dd^\top \dd_i} \le 1 - \frac{1}{10k}
\right)
\ge \frac{4}{5}.
\]
This implies there exists a $\dd$ as desired.
\end{proof}

We independently pick $O(\log n)$ unit vectors uniformly at random.
By a Chernoff bound, we can find a direction $\dd$ satisfying Equation~\eqref{eqn:directionD} with high probability.

Recall that Lemma~\ref{lem:planeIntersect} states that: 
 any 2D plane $P$ orthogonal to $\dd$ intersects a small number of tetrahedrons in a $(B, r)$-hollowing of a {\trussnice} 3D truss.
Lemma~\ref{lem:smallSeparatorPlane} 
bounds the number of tetrahedrons in a {\trussnice} 3D truss which intersect a single 2D plane $P$ orthogonal to $\dd$.
We restate it in the following, which can be proved by combining Observation~\ref{obs:convertToSurfaceArea} and Claim~\ref{clm:intersectArea}.

\begin{replemma}{lem:smallSeparatorPlane} 
Let $\calT$ be a {\trussnice} 3D truss of $n$ vertices, and let $B$ be a bounding box of $\calT$. 
Let $\dd \in \mathbb{R}^3$ be a unit vector such that the angle between $\dd$ and the longest direction of $B$ is $\theta \neq \pi /2$.
Then any plane $P$ orthogonal to $\dd$ intersects $\calT$ in at most $O(n^{2/3} \alpha(\calT)^{-1/3} \cos^{-1} \theta)$ tetrahedrons.
\end{replemma}

Lemma~\ref{lem:planeIntersect} and Lemma~\ref{lem:smallSeparatorPlane} together imply that:
the top-level separator $S_1 \cup \ldots \cup S_l$ computed in Algorithm~\ref{alg:NDkChunks} has a small size.
This, together with nested dissection in~\cite{millerT90}, lets us prove that Algorithm~\ref{alg:NDkChunks} outputs an elimination ordering with a small fill-in size and multiplication count.

\begin{lemma}
\label{lem:kChunkNDMoreDetailed}
Given a {\esimple} 3D truss $\calT$ of $n$ vertices which
is a union of $k$ {\trussnice} trusses with $n_i$ vertices each,
running Algorithm~\ref{alg:TrussSolver},
$\textsc{TrussSolver}(\calT, \ff, \epsilon, c_\alpha, c_r)$,
with Line~\ref{line:MainND} replaced by
Algorithm~\ref{alg:NDkChunks}, $\algKChunksND(\calT, \calB, \calI, \calH, l)$,
leads to performance in terms of $n$ that is optimized by setting
\[
c_{\alpha} =  c_r \le \frac{1}{3}
\]
in Line~\ref{line:MainI} of Algorithm~\ref{alg:TrussSolver}, \textsc{TrussSolver}.
In terms of $c_r$, the hollowing parameter,
and $l$, the number of top-level separators,
this gives an elimination ordering with fill-in size at most
\[
O(n^{4/3} l^{-1/3}  + k^{14/3 + 2c_r/3} n^{4/3 - 2c_r/3} l),
\]
that can be computed in time
\[
O(n^{2\omega / 3} l^{-2\omega / 3 + 1} 
+ k^{7 \omega / 3 + \omega c_r / 3} n^{2\omega / 3 - c_r\omega / 3} l \log n),
\]
where $\omega$ is the matrix multiplication exponent.
\end{lemma}

\begin{proof}
We apply Lemma~\ref{lem:flexibleND}. 
According to Algorithm~\ref{alg:NDkChunks} line~\ref{lin:directionD}, for each $i$, the angle between the longest direction of the $i$th bounding box and $\dd$ has cosine value in $[1/10k, 1-1/10k]$.

We first upper bound the number of vertices in each top-level separators, that is, the number of tetrahedrons in $(\union_{i \in \calI} \calH_i) \bigcup (\union_{i \notin \calI} \calT_i)$ which intersects a plane $P_j$, see Algorithm~\ref{alg:NDkChunks} line~\ref{lin:topSeparator}.
For each $i \in \calI \defeq \{i \in [k]: \alpha(\calT_i) \le n_i^{c_\alpha}\}$, by Lemma~\ref{lem:planeIntersect}, the number of tetrahedrons in $\calH_i$ intersect a plane $P_j$ is at most 
\[
O \left( k^2 n_i^{2/3}  \alpha_i^{-1/3}  r_i^{-1/3} \right)
= O \left( k^2 n_i^{2/3 - c_r/3} \right),
\]
since $\alpha_i \ge 1$.
For each $i \notin \calI$, by Lemma~\ref{lem:smallSeparatorPlane},
the number of tetrahedrons in $\calT_i$ intersect a plane $P_j$ is at most 
\[
O \left( k n_i^{2/3}  \alpha_i^{-1/3}  \right)
=
O \left( k n_i^{2/3 - c_{\alpha} /3}  \right)
\]
since $\alpha_i > n_{i}^{c_{\alpha}}$.
The two terms have same exponent for $n$ when we set 
$c_r = c_{\alpha}$.
Note Algorithm~\ref{alg:hollow} requires that $c_r + 2 c_{\alpha} \le 1$. Thus, here we need $c_r \le 1/3$.

Thus, the total number of tetrahedrons in
$(\union_{i \in \calI} \calH_i) \bigcup (\union_{i \notin \calI} \calT_i)$
which intersect a single separator plane $P_j$ is then at most:
\[
s \defeq O \left( \sum_{1\le i \le k} k^2 n_i^{2/3 - c_r/3}  \right)
= O \left( k^{7/3 + c_r/3} n^{2/3 - c_r/3} \right).
\]
The last inequality is by Jensen's inequality.

There are totally $l$ top-level separators,  which separates the whole truss into $l+1$ separate components and each component has $O(n/l)$ vertices, according to Algorithm~\ref{alg:NDkChunks} line~\ref{lin:separatorPlane}.

We plug these parameters into Lemma~\ref{lem:flexibleND}, the total fill-in size is 
\[
O \left( \left( \frac{n}{l} \right)^{4/3} l  + s^2 l \right)
= O (n^{4/3} l^{-1/3}  + k^{14/3 + 2c_r/3} n^{4/3 - 2c_r/3} l  ),
\]
and the multiplication count is
\[
O \left( \left( \frac{n}{l} \right)^{2 \omega /3} l
+ s^{\omega} l \log n
\right)
= O (n^{2\omega / 3} l^{-2\omega / 3 + 1} 
+ k^{7 \omega / 3 + \omega c_r / 3} n^{2\omega / 3 - c_r\omega / 3} l  \log n).
\]
This completes the proof.
\end{proof}

\bibliographystyle{alpha}
\bibliography{ref}

\newcommand{\etalchar}[1]{$^{#1}$}
\begin{thebibliography}{BENWN14}

\bibitem[Axe85]{axelsson85}
Owe Axelsson.
\newblock A survey of preconditioned iterative methods for linear systems of
  algebraic equations.
\newblock {\em BIT Numerical Mathematics}, 25(1):165--187, 1985.

\bibitem[AY10]{alonY10}
Noga Alon and Raphael Yuster.
\newblock Solving linear systems through nested dissection.
\newblock In {\em Foundations of Computer Science (FOCS), 2010 51st Annual IEEE
  Symposium on}, pages 225--234. IEEE, 2010.

\bibitem[BCFN16]{borradaileCFN16}
Glencora Borradaile, Erin~Wolf Chambers, Kyle Fox, and Amir Nayyeri.
\newblock Minimum cycle and homology bases of surface embedded graphs.
\newblock {\em arXiv preprint arXiv:1607.05112}, 2016.

\bibitem[BEG94]{bernEG94}
Marshall Bern, David Eppstein, and John Gilbert.
\newblock Provably good mesh generation.
\newblock {\em Journal of Computer and System Sciences}, 48(3):384--409, 1994.

\bibitem[BENWN14]{borradaileENW14}
Glencora Borradaile, David Eppstein, Amir Nayyeri, and Christian Wulff-Nilsen.
\newblock All-pairs minimum cuts in near-linear time for surface-embedded
  graphs.
\newblock {\em arXiv preprint arXiv:1411.7055}, 2014.

\bibitem[BHP01]{BarequetH01}
Gill Barequet and Sariel Har-Peled.
\newblock Efficiently approximating the minimum-volume bounding box of a point
  set in three dimensions.
\newblock {\em J. Algorithms}, 38(1):91--109, January 2001.

\bibitem[BKM{\etalchar{+}}11a]{borradaileKMNW11}
Glencora Borradaile, Philip~N Klein, Shay Mozes, Yahav Nussbaum, and Christian
  Wulff-Nilsen.
\newblock Multiple-source multiple-sink maximum flow in directed planar graphs
  in near-linear time.
\newblock In {\em Foundations of Computer Science (FOCS), 2011 IEEE 52nd Annual
  Symposium on}, pages 170--179. IEEE, 2011.

\bibitem[BKM{\etalchar{+}}11b]{BorradailePMNW11}
Glencora Borradaile, Philip~N Klein, Shay Mozes, Yahav Nussbaum, and Christian
  Wulff-Nilsen.
\newblock Multiple-source multiple-sink maximum flow in directed planar graphs
  in near-linear time.
\newblock In {\em Foundations of Computer Science (FOCS), 2011 IEEE 52nd Annual
  Symposium on}, pages 170--179. IEEE, 2011.
\newblock Available at: https://arxiv.org/abs/1105.2228.

\bibitem[BSS13]{bandeiraSS13}
Afonso~S Bandeira, Amit Singer, and Daniel~A Spielman.
\newblock A cheeger inequality for the graph connection laplacian.
\newblock {\em SIAM Journal on Matrix Analysis and Applications},
  34(4):1611--1630, 2013.

\bibitem[CFM{\etalchar{+}}14]{CohenFMNPW14}
Michael~B. Cohen, Brittany~Terese Fasy, Gary~L. Miller, Amir Nayyeri, Richard
  Peng, and Noel Walkington.
\newblock Solving 1-laplacians in nearly linear time: Collapsing and expanding
  a topological ball.
\newblock In {\em Proceedings of the Twenty-Fifth Annual {ACM-SIAM} Symposium
  on Discrete Algorithms, {SODA} 2014, Portland, Oregon, USA, January 5-7,
  2014}, pages 204--216, 2014.
\newblock Available at: https://www.cs.cmu.edu/\textasciitilde
  glmiller/Publications/Papers/CoFMNPW14.pdf.

\bibitem[Che89]{chew89}
L~Paul Chew.
\newblock Guaranteed-quality triangular meshes.
\newblock Technical report, Cornell University, 1989.

\bibitem[Chu96]{chung96}
Fan~RK Chung.
\newblock Laplacians of graphs and cheeger’s inequalities.
\newblock {\em Combinatorics, Paul Erdos is Eighty}, 2(157-172):13--2, 1996.

\bibitem[CKM{\etalchar{+}}14]{CKMPPRX14}
Michael~B Cohen, Rasmus Kyng, Gary~L Miller, Jakub~W Pachocki, Richard Peng,
  Anup~B Rao, and Shen~Chen Xu.
\newblock Solving sdd linear systems in nearly m log 1/2 n time.
\newblock In {\em Proceedings of the 46th Annual ACM Symposium on Theory of
  Computing}, pages 343--352. ACM, 2014.

\bibitem[CKP{\etalchar{+}}17]{CohenKPPRSV16}
Michael~B. Cohen, Jonathan Kelner, John Peebles, Richard Peng, Anup~B. Rao,
  Aaron Sidford, and Adrian Vladu.
\newblock Almost-linear-time algorithms for markov chains and new spectral
  primitives for directed graphs.
\newblock In {\em Proceedings of the 49th Annual ACM SIGACT Symposium on Theory
  of Computing}, STOC 2017, pages 410--419, New York, NY, USA, 2017. ACM.

\bibitem[DS07]{daitchS07}
Samuel~I Daitch and Daniel~A Spielman.
\newblock Support-graph preconditioners for 2-dimensional trusses.
\newblock {\em arXiv preprint cs/0703119}, 2007.

\bibitem[EN11]{ericksonN11}
Jeff Erickson and Amir Nayyeri.
\newblock Computing replacement paths in surface embedded graphs.
\newblock In {\em Proceedings of the twenty-second annual ACM-SIAM symposium on
  Discrete Algorithms}, pages 1347--1354. Society for Industrial and Applied
  Mathematics, 2011.

\bibitem[Fed64]{fedorenko64}
Radii~Petrovich Fedorenko.
\newblock The speed of convergence of one iterative process.
\newblock {\em USSR Computational Mathematics and Mathematical Physics},
  4(3):227--235, 1964.

\bibitem[Fre87]{frederickson87}
Greg~N. Frederickson.
\newblock Fast algorithms for shortest paths in planar graphs, with
  applications.
\newblock {\em {SIAM} J. Comput.}, 16(6):1004--1022, 1987.

\bibitem[Geo73]{george73}
Alan George.
\newblock Nested dissection of a regular finite element mesh.
\newblock {\em SIAM Journal on Numerical Analysis}, 10(2):345--363, 1973.

\bibitem[GNP94]{gilbertNP94}
John~R Gilbert, Esmond~G Ng, and Barry~W Peyton.
\newblock An efficient algorithm to compute row and column counts for sparse
  cholesky factorization.
\newblock {\em SIAM Journal on Matrix Analysis and Applications},
  15(4):1075--1091, 1994.

\bibitem[Goo95]{goodrich95}
Michael~T Goodrich.
\newblock Planar separators and parallel polygon triangulation.
\newblock {\em Journal of Computer and System Sciences}, 51(3):374--389, 1995.

\bibitem[HKRS97]{henzingerKRS97}
Monika~R Henzinger, Philip Klein, Satish Rao, and Sairam Subramanian.
\newblock Faster shortest-path algorithms for planar graphs.
\newblock {\em journal of computer and system sciences}, 55(1):3--23, 1997.

\bibitem[KLP{\etalchar{+}}16]{KLPSS16}
Rasmus Kyng, Yin~Tat Lee, Richard Peng, Sushant Sachdeva, and Daniel~A.
  Spielman.
\newblock Sparsified cholesky and multigrid solvers for connection laplacians.
\newblock In {\em Proceedings of the Forty-eighth Annual ACM Symposium on
  Theory of Computing}, STOC '16, pages 842--850, New York, NY, USA, 2016. ACM.

\bibitem[KMP10]{KoutisMP10}
Ioannis Koutis, Gary~L. Miller, and Richard Peng.
\newblock Approaching optimality for solving {SDD} linear systems.
\newblock In {\em Proceedings of the 2010 IEEE 51st Annual Symposium on
  Foundations of Computer Science}, FOCS '10, pages 235--244, Washington, DC,
  USA, 2010. IEEE Computer Society.
\newblock Available at http://arxiv.org/abs/1003.2958.

\bibitem[KMP11]{KoutisMP11}
Ioannis Koutis, Gary~L. Miller, and Richard Peng.
\newblock A nearly-m log n time solver for {SDD} linear systems.
\newblock In {\em Proceedings of the 2011 IEEE 52nd Annual Symposium on
  Foundations of Computer Science}, FOCS '11, pages 590--598, Washington, DC,
  USA, 2011. IEEE Computer Society.
\newblock Available at http://arxiv.org/abs/1102.4842.

\bibitem[KMS13]{kleinMS13}
Philip~N Klein, Shay Mozes, and Christian Sommer.
\newblock Structured recursive separator decompositions for planar graphs in
  linear time.
\newblock In {\em Proceedings of the forty-fifth annual ACM symposium on Theory
  of computing}, pages 505--514. ACM, 2013.

\bibitem[KS16]{kyngS16}
Rasmus Kyng and Sushant Sachdeva.
\newblock Approximate gaussian elimination for laplacians-fast, sparse, and
  simple.
\newblock In {\em Foundations of Computer Science (FOCS), 2016 IEEE 57th Annual
  Symposium on}, pages 573--582. IEEE, 2016.

\bibitem[KZ17]{kyngZ17}
Rasmus Kyng and Peng Zhang.
\newblock Hardness results for structured linear systems.
\newblock {\em arXiv preprint arXiv:1705.02944}, 2017.

\bibitem[LG14]{Legall14}
Fran{\c{c}}ois Le~Gall.
\newblock Powers of tensors and fast matrix multiplication.
\newblock In {\em Proceedings of the 39th international symposium on symbolic
  and algebraic computation}, pages 296--303. ACM, 2014.
\newblock Available at: https://arxiv.org/abs/1401.7714.

\bibitem[LGT14]{leeGT14}
James~R Lee, Shayan~Oveis Gharan, and Luca Trevisan.
\newblock Multiway spectral partitioning and higher-order cheeger inequalities.
\newblock {\em Journal of the ACM (JACM)}, 61(6):37, 2014.

\bibitem[LRT79]{liptonRT79}
Richard~J Lipton, Donald~J Rose, and Robert~Endre Tarjan.
\newblock Generalized nested dissection.
\newblock {\em SIAM journal on numerical analysis}, 16(2):346--358, 1979.

\bibitem[{\L}S11]{lkackiS11}
Jakub {\L}{\k{a}}cki and Piotr Sankowski.
\newblock Min-cuts and shortest cycles in planar graphs in o (n loglogn) time.
\newblock In {\em European Symposium on Algorithms}, pages 155--166. Springer,
  2011.

\bibitem[MT90]{millerT90}
Gary~L Miller and William Thurston.
\newblock Separators in two and three dimensions.
\newblock In {\em Proceedings of the twenty-second annual ACM symposium on
  Theory of computing}, pages 300--309. ACM, 1990.

\bibitem[MTTV98]{millerTTV98}
Gary~L Miller, Shang-Hua Teng, William Thurston, and Stephen~A Vavasis.
\newblock Geometric separators for finite-element meshes.
\newblock {\em SIAM Journal on Scientific Computing}, 19(2):364--386, 1998.

\bibitem[MV92]{mitchellV92}
Scott~A Mitchell and Stephen~A Vavasis.
\newblock Quality mesh generation in three dimensions.
\newblock In {\em Proceedings of the eighth annual symposium on Computational
  geometry}, pages 212--221. ACM, 1992.

\bibitem[PRT16]{parzanchevski2016isoperimetric}
Ori Parzanchevski, Ron Rosenthal, and Ran~J Tessler.
\newblock Isoperimetric inequalities in simplicial complexes.
\newblock {\em Combinatorica}, 36(2):195--227, 2016.

\bibitem[PS14]{pengS14}
Richard Peng and Daniel~A Spielman.
\newblock An efficient parallel solver for sdd linear systems.
\newblock In {\em Proceedings of the forty-sixth annual ACM symposium on Theory
  of computing}, pages 333--342. ACM, 2014.

\bibitem[RTL76]{roseTL76}
Donald~J Rose, R~Endre Tarjan, and George~S Lueker.
\newblock Algorithmic aspects of vertex elimination on graphs.
\newblock {\em SIAM Journal on computing}, 5(2):266--283, 1976.

\bibitem[Rup93]{ruppert93}
Jim Ruppert.
\newblock A new and simple algorithm for quality 2-dimensional mesh generation.
\newblock In {\em Proceedings of the fourth annual ACM-SIAM Symposium on
  Discrete algorithms}, pages 83--92. Society for Industrial and Applied
  Mathematics, 1993.

\bibitem[Saa03]{saad03}
Yousef Saad.
\newblock {\em Iterative methods for sparse linear systems}.
\newblock SIAM, 2003.

\bibitem[SF73]{strangF73}
Gilbert Strang and George~J Fix.
\newblock {\em An analysis of the finite element method}, volume 212.
\newblock Prentice-hall Englewood Cliffs, NJ, 1973.

\bibitem[SKM14]{steenbergenKM14}
John Steenbergen, Caroline Klivans, and Sayan Mukherjee.
\newblock A cheeger-type inequality on simplicial complexes.
\newblock {\em Advances in Applied Mathematics}, 56:56--77, 2014.

\bibitem[ST08]{shklarskiT08}
Gil Shklarski and Sivan Toledo.
\newblock Rigidity in finite-element matrices: Sufficient conditions for the
  rigidity of structures and substructures.
\newblock {\em SIAM Journal on Matrix Analysis and Applications}, 30(1):7--40,
  2008.

\bibitem[ST14]{spielmanT14}
Daniel~A Spielman and Shang-Hua Teng.
\newblock Nearly linear time algorithms for preconditioning and solving
  symmetric, diagonally dominant linear systems.
\newblock {\em SIAM Journal on Matrix Analysis and Applications},
  35(3):835--885, 2014.

\end{thebibliography}

\appendix
\section*{Appendices}

\section{Schur Complements}
\label{sec:appendixSchur}

In this section, we prove Fact~\ref{fac:Schur} on Schur complements.
Throughout this subsection, let
\[
\AA = \left( \begin{array}{cc}
\AA_{SS} & \AA_{ST} \\
\AA_{ST}^\top & \AA_{TT}
\end{array} \right)
\]
be a symmetric matrix.
Let $n_1$ be the size of $S$ and $n_2$ be the size of $T$.
Recall that the Schur complement
\[
\textsc{Sc}[\AA]_T = 
\AA_{TT} - \AA_{ST}^\top\AA^{\dagger}_{SS}\AA_{ST}.
\]

We restate Fact~\ref{fac:Schur} below.

\begin{repfact}{fac:Schur}
Let $\AA$ be a symmetric PSD matrix defined as above, 
and let $\schurto{\AA}{T}$ be its Schur complement.
Then,
\begin{enumerate}
\item $\schurto{\AA}{T}$ is a symmetric PSD matrix.
\item $\lambda_{\max}(\schurto{\AA}{T}) \le \lambda_{\max}(\AA)$.
\end{enumerate}
\end{repfact}

To prove this fact, we need the following fact and a special case of Weyl inequalities.

\begin{fact}

For any fixed vector $\yy \in \mathbb{R}^{n_2}$, 
\[
\min_{\xx \in \mathbb{R}^{n_1}} \left( \begin{array}{cc}
\xx^{\trp} & \yy^{\trp}
\end{array} \right) \AA \left( \begin{array}{c}
\xx \\
\yy
\end{array} \right) 
= \yy^{\trp} \textsc{Sc}[\AA]_T \yy.
\]
\label{clm:schur}
\end{fact}
\begin{proof}
We expand the left hand side,
\begin{align}
\left( \begin{array}{cc}
\xx^{\trp} & \yy^{\trp}
\end{array} \right) \AA \left( \begin{array}{c}
\xx \\
\yy
\end{array} \right) 
= \xx^{\trp}\AA_{SS}\xx
+ 2\xx^{\trp} \AA_{ST}\yy + \yy^{\trp} \AA_{TT} \yy.
\label{eqn:schur_lhs}
\end{align}
Taking derivative w.r.t. $\xx$ and setting it to be 0 give that
\[
2 \AA_{SS}\xx + 2 \AA_{ST} \yy = {\bf 0}.
\]
Plugging $\xx = - \AA_{SS}^{\dagger} \AA_{ST} \yy$ into~\eqref{eqn:schur_lhs}, 
\[
\min_{\xx \in \mathbb{R}^{n_1}} \left( \begin{array}{cc}
\xx^{\trp} & \yy^{\trp}
\end{array} \right) \AA \left( \begin{array}{c}
\xx \\
\yy
\end{array} \right) 
= -  \yy^\top \AA_{ST}^\top \AA^{\dagger}_{SS} \AA_{ST} \yy 
+ \yy^\top \AA_{TT} \yy
= \yy^{\trp} \textsc{Sc}[\AA]_T \yy.
\]
This completes the proof.
\end{proof}

\begin{theorem}[A special case of Weyl inequalities]
Let $\HH \in \mathbb{R}^n$ and $\HH = \HH_1 + \HH_2$ where $\HH_1, \HH_2$ are symmetric matrices and $\HH_2$ is a PSD matrix. 
Let $\lambda_1(\cdot) \ge \ldots \ge \lambda_n(\cdot)$ be eigenvalues of a matrix.
Then, $\lambda_i (\HH) \ge \lambda_i (\HH_1), i = 1,2,\ldots, n$.
\label{thm:weyl}
\end{theorem}

\begin{proof}[Proof of Fact~\ref{fac:Schur}]

The first statement immediately follows Fact~\ref{clm:schur}.

To prove the second statement,
we decompose $\AA$:
\begin{align*}
\AA & = \left( \begin{array}{cc}
{\bf 0} & {\bf 0} \\ 
{\bf 0} & \textsc{Sc}[\AA]_T
\end{array} \right)
+ \left( \begin{array}{cc}
\AA_{SS} & \AA_{ST} \\
\AA_{ST}^\top &  \AA_{ST}^\top \AA_{SS}^{-1} \AA_{ST}
\end{array} \right) \\
& = \left( \begin{array}{cc}
{\bf 0} & {\bf 0} \\ 
{\bf 0} & \textsc{Sc}[\AA]_T
\end{array} \right)
+ \left( \begin{array}{c}
\AA_{SS}^{1/2} \\
\AA_{ST}^\top \AA_{SS}^{-1/2} 
\end{array} \right)^\top
\left( \begin{array}{c}
\AA_{SS}^{1/2} \\
\AA_{ST}^\top \AA_{SS}^{ -1/2}
\end{array} \right).
\end{align*}
Here we assume that $\AA_{SS}$ is invertible, otherwise we can use pseudo-inverse.
The first matrix is symmetric, and the second matrix is symmetric and PSD.
By Theorem~\ref{thm:weyl}, 
$\lambda_{\max}(\AA) \ge \lambda_{\max}(\textsc{Sc}[\AA]_T)$.
\end{proof}

%%% Local Variables:
%%% mode: latex
%%% TeX-master: "main"
%%% End:

\section{Running Times In Terms of Fast Matrix Multiplication}
\label{sec:rofl}

We now restate the running times of our algorithms in terms of
faster matrix multiplication / inversion routines.
Specifically, we assume inverting an $n \times n$ matrix takes
time $O(n^{\omega})$, where $\omega < 2.3728639$~\cite{Legall14}.

We first examine purely nested dissection based algorithms.
The running time of these algorithms are dominated by
the cost of inverting the matrix at the top-most level.
Thus, the running time of 3-D nested dissection from~\cite{millerT90}
as given in Theorem~\ref{thm:nestedDissection} is
\[
O\left( n^{2 \omega / 3}\right),
\]
while the performance of Lemma~\ref{lem:flexibleND} becomes
\[
O\left( n^{2\omega \beta / 3 + \alpha} 
+ n^{\omega \gamma  +\alpha}
\right).
\]

We now propagate these different costs for constructing the
nested dissection partial states into our running time analyses.

For the bounded aspect ratio case described in
Theorem~\ref{thm:mainSmallAR}, recall that the input truss is a union of $k$ convex pieces of $n_i$ vertices each.
Hollowing with parameter
\[
r_i = n_i^{c_r}
\]
now takes time
\[
O\left(
\sum_{i}
\frac{n_{i}}{r_{i}} \cdot r_{i}^{2 \omega / 3}
\right)
=
O\left(
n^{1 + c_r \left( 2 \omega / 3 - 1 \right)}
\right),
\]
while still giving a Schur complement of size
\[
O\left( n^{1 + c_r / 3} \right)
\]
on the boundaries, and total boundary size of
\[
O\left( k^{c_r/3} n^{1 - c_r / 3} \right).
\]

Theorem~\ref{thm:nestedDissection} then gives that solving this
problem on just the boundary elements takes time
\[
O\left( k^{2c_r \omega /9} n^{2\left( 1 - c_r /3 \right)\omega / 3} \right)
=
O\left( k^{2c_r \omega /9} n^{2 \omega / 3 - 2 c_r \omega /9} \right)
,
\]
and results in a total fill-in of
\[
O\left( k^{4 c_r/9} n^{4/3 - 4 c_r/ 9} \right).
\]

Putting these parameters back into the condition number
bound of $O(n^{2 c_r})$ gives an iteration count of $O(n^{c_r} \log(1 / \epsilon))$,
which in turn gives a total cost of
\begin{multline*}
O\left(
n^{1 + c_r \left( 2 \omega / 3 - 1 \right)}
\right)
+
O\left( k^{2c_r \omega /9} n^{2 \omega / 3 - 2 c_r \omega /9} \right)
+ O\left( n^{c_r} \log\left(1 / \epsilon\right)  \right)
\cdot
\left[
O\left( n^{1 + c_r / 3} \right)
+
O\left( k^{c_r/3} n^{4/3 - 4 c_r/ 9} \right)
\right]
\\
=
O\left(
n^{1 + c_r \left( 2 \omega / 3 - 1 \right)}
+
k^{2c_r \omega /9} n^{2 \omega / 3 - 2 c_r \omega /9}
+
n^{1 + 4 c_r / 3} \log\left(1 / \epsilon\right)
+
k^{c_r/3} n^{4/3 + 5 c_r/ 9} \log\left(1 / \epsilon\right)
\right).
\end{multline*}
We can (slightly) simplify this using the fact that $\omega \le 3$
to drop the first term: it is always upper bounded by the third.
Also, since $k \leq n$, we will focus on optimizing the
exponent on $n$, that is, we want to pick $c_r$ to minimize
the maximum of
\begin{align*}
& 2 \omega /3 - 2 c_r \omega / 9\\
& 1 + 4 c_r / 3\\
& 4/3 + 5c_r / 9
\end{align*}

By running an LP solver, we get:
\begin{itemize}
\item
 when $\omega = 3$ this is optimized at
$c_r = 1/2$,
which gives a total cost of
$O(k^{1/3} n^{5/3} \log( 1 / \epsilon))$.
\item when $\omega = 2.3728639$, 
this is optimized at $c_r = 0.2295553$.
Here the exponents on $n$ the three terms are
$1.4608640$, $1.3060737$ and $1.4608640$
respectively, and we have
$2 \omega / 9 > 1 /3$, so 
so the total cost is bounded by
$O(k^{0.1210452} n^{1.4608641} \log( 1 / \epsilon))$.
\end{itemize}

For the more general case from Theorem~\ref{thm:main},
combining the bounds from Lemma~\ref{lem:kChunkNDMoreDetailed}
with the
\begin{itemize}
\item $O(n^{1 + c_r ( 2 \omega / 3 - 1 )})$ cost of computing the
Schur complement of eliminating the innards of the hollowings, and
\item the $O( n^{1 + c_r / 3} )$ size of the these Schur complements, and
\item $O(n^{c_r} \log(1 / \epsilon))$ iteration count of PCG
\end{itemize}
gives a total cost of:\footnote{We drop the $\log n$ factor here, given $\log n \ll n^c$ for any constant $c> 0$.}
\begin{multline*}
O\left(
n^{1 + c_r \left( 2 \omega / 3 - 1 \right)}
+
n^{2\omega / 3} l^{-2\omega / 3 + 1} 
+ k^{7 \omega / 3 + c_r\omega / 3} n^{2\omega / 3 - c_r\omega / 3} l\right. \\ \left.
+ n^{c_r} \log\left( 1 / \epsilon \right)
\left(
  n^{1 + c_r / 3} 
  +
  n^{4/3} l^{-1/3} + k^{14/3+2c_r/3} n^{4/3 - 2c_r/3} l
\right)
\right).
\end{multline*}
Since $\omega \le 3$, we drop the first term.
We can simplify this by moving $k$ and $\log(1 / \epsilon)$ to the outermost,
and only optimizing the remaining terms:
\[
O\left( k^{7 \omega / 3 + c_r\omega / 3} \log\left( 1 / \epsilon \right) \right)
\cdot 
\left( n^{2\omega / 3} l^{-2\omega / 3 + 1} 
+  n^{2\omega / 3 - c_r\omega / 3} l
+ \left(
  n^{1 + 4c_r / 3} 
  +
  n^{4/3 + c_r} l^{-1/3} +  n^{4/3 + c_r/3} l
\right)
\right).
\]

Let $l = n^{c_l}$.
Since Algorithm~\ref{alg:hollow} requires that $c_r +  2c_{\alpha} \le 1$ and in Lemma~\ref{lem:kChunkNDMoreDetailed} we set $c_r = c_{\alpha}$, we have $0 \le c_r \le 1/3$.
Subject to this requirement, 
we minimize the maximum of the following terms:
\begin{align*}
& 1 + \left( \frac{2 \omega}{3} - 1 \right) c_r \\
& \frac{2\omega}{3} + \left(-\frac{2\omega}{3} + 1\right) c_l , \\
& \frac{2 \omega}{3} - \frac{\omega}{3} c_r  +  c_l, \\
& 1 + \frac{4}{3} c_r, \\
& \frac{4}{3} + c_r - \frac{1}{3} c_l, \\
& \frac{4}{3} +  \frac{1}{3} c_r + c_l.
\end{align*}

By running an LP solver, we get:
\begin{itemize}
\item
 when $\omega = 3$ this is optimized at
$c_r = 1/3$ and $c_l = 1/6$,
which gives a total cost of $O(k^{22/3} n^{11/6} \log (1 / \eps))$ ($11/6 \approx 1.8333$)
\item when $\omega = 2.3728639$, an optimum solution is $c_r = 0.2210963$
and $c_{l} = 0.1105482$
for a total cost of $O(k^{5.7115596} n^{1.5175803} \log(1/\eps))$.
\end{itemize}

\end{document}